\documentclass[11pt]{article}

\usepackage{fullpage,amsmath,xspace,amssymb,epsfig,graphicx,cite}
\usepackage[usenames,dvipsnames,dvips]{color}

\newtheorem{Thm}{Theorem}[section]
\newtheorem{Lem}[Thm]{Lemma}
\newtheorem{Cor}[Thm]{Corollary}
\newtheorem{Prop}[Thm]{Proposition}

\newtheorem{Def}[Thm]{Definition}
\newtheorem{Fact}{Fact}
\newtheorem{Conj}{Conjecture}
\newenvironment{proof}{\noindent {\textbf{Proof }}}{$\Box$ \medskip}

\newcommand\mbN{\mbox{$\mathbb{N}$}}

\newcommand\mbR{\mbox{$\mathbb{R}$}}

\newcommand\mcX{\mathcal{X}}
\newcommand\mcY{\mathcal{Y}}
\newcommand\mcQ{\mathcal{Q}}
\newcommand\mcM{\mathcal{M}}
\newcommand\B{\{0,1\}}

\newcommand {\ie} {\textit{i.e.}\xspace}
\newcommand {\st} {\textit{s.t.}\xspace}

\newcommand\pr{\mbox{\bf Pr}}
\newcommand\av{\mbox{\bf E}}
\newcommand\ppad{\mbox{\bf PPAD}\xspace}

\newcommand\cptp{\mbox{\sf {CPTP}}\xspace}

\newcommand\alice{\mbox{\sf {Alice}}\xspace}
\newcommand\bob{\mbox{\sf {Bob}}\xspace}
\newcommand\rcomm{\mbox{\sf {RComm}}\xspace}
\newcommand\qcomm{\mbox{\sf {QComm}}\xspace}
\newcommand\rcorr{\mbox{\sf {RCorr}}\xspace}
\newcommand\qcorr{\mbox{\sf {QCorr}}\xspace}
\newcommand\size{\mbox{\tt {size}}\xspace}
\newcommand\rank{\mbox{\tt {rank}}\xspace}
\newcommand {\good} {\text{good}}

\newcommand\ket[1]{| #1 \rangle}
\newcommand\bra[1]{\langle #1 |}
\newcommand\qip[2]{\langle #1 | #2 \rangle}

\begin{document}
\title{\bf Quantum Strategic Game Theory}
\author{Shengyu Zhang\thanks{Department of Computer Science and Engineering and The Institute of Theoretical Computer Science and Communications, The Chinese University of Hong Kong, Shatin, N.T., Hong Kong. Email: {\tt syzhang@cse.cuhk.edu.hk}.}}
\date{}
\maketitle

\begin{abstract}
We propose a simple yet rich model to extend strategic games to the quantum setting, in which we define quantum Nash and correlated equilibria and study the relations between classical and quantum equilibria. Unlike all previous work that focused on qualitative questions on specific games of very small sizes, we quantitatively address the following fundamental question for general games of growing sizes: 
\begin{quote}
	\emph{How much ``advantage" can playing quantum strategies provide, if any?} 
\end{quote}
Two measures of the advantage are studied.
\begin{enumerate}
	\item A natural measure is the increase of payoff. We consider natural mappings between classical and quantum states, and study how well those mappings preserve the equilibrium properties. Among other results, we exhibit a correlated equilibrium $p$ whose quantum superposition counterpart $\sum_s \sqrt{p(s)}\ket{s}$ is far from being a quantum correlated equilibrium; actually a player can increase her payoff from almost 0 to almost 1 in a $[0,1]$-normalized game. We achieve this by a tensor product construction on carefully designed base cases. 
	
	\item Another measure is the hardness of generating correlated equilibria, for which we propose to study \emph{correlation complexity}, a new complexity measure for correlation generation. We show that there are $n$-bit correlated equilibria which can be generated by only one EPR pair followed by local operation (without communication), but need at least $\log_2(n)$ classical shared random bits plus communication. The randomized lower bound can be improved to $n$, the best possible, assuming (even a much weaker version of) a recent conjecture in linear algebra. We believe that the correlation complexity, as a complexity-theoretical counterpart of the celebrated Bell's inequality, has independent interest in both physics and computational complexity theory and deserves more explorations. 
\end{enumerate}
\end{abstract}

\section{Introduction}\label{sec: intro}
\subsection{Game theory}
Game theory is a branch of applied mathematics to model and analyze interactions of two or more individuals, usually called \emph{players}, each with a possibly different goal. Over decades of development, game theory has grown into a rich field, and has found numerous applications in economics, political science, biology, philosophy, statistics, computer science, etc. Many models have been proposed to study games, among which the most popular and fundamental ones are \emph{strategic games} (or games in strategic or normal form) and \emph{extensive games} (or games in extensive form). In the former, the players choose their strategies simultaneously, and then each receives a payoff based on all players' strategies. In the latter the players choose their strategies adaptively in turn, and finally when all players finish their moves, each receives a payoff based on the entire history of moves of all players. Variation in settings exists. For instance, if before playing the game, each player also receives a private and random input, then they are playing a \emph{Bayesian game}, which belongs to the larger class of \emph{games with incomplete information}. See standard textbooks such as \cite{OR94,FT91} for more details. 

Motivated by the emergence of Internet and other systems with a huge number of players, various algorithmic and complexity-theoretical perspectives from computer science have been added as one more dimension for studying games. %Central concepts such as equilibria and important areas such as mechanism design have been revisited with algorithmic ingredients injected.
See an excellent recent textbook \cite{VNRT07} for more background on this emerging field of \emph{algorithmic game theory}.

Equilibrium as a central solution concept in game theory attempts to capture the situation in which each player has adopted an optimal strategy, provided that others keep their strategies unchanged. Nash equilibrium \cite{vNM44,Nas50,Nas51}\footnote{Introduced by von Neumann and Morgenstern \cite{vNM44} who showed existence of a Nash equilibrium in any zero-sum game, existence later extended by Nash, a Laureate of Nobel Prize in Economic Sciences, to any game with a finite set of strategies \cite{Nas51}.} is the first and most fundamental concept of equilibrium. A joint strategy is a pure Nash equilibrium if no player has any incentive to change her strategy. If each player draws her strategies from a probability distribution, and no player can increase her expected payoff by switching to any other strategy on average of other players' strategies, then they are playing a mixed Nash equilibrium. Note that here we require no correlation between players' probabilistic strategies. 

One important extension of Nash equilibrium is \emph{correlated equilibrium} \cite{Aum74}\footnote{Defined by Aumann, another Laureate of Nobel Prize in Economic Sciences.}, which relaxes the above independence requirement. We can think of a correlated equilibrium being generated by a Referee (or a ``Mediator"), who samples a joint strategy from the correlated distribution and sends the $i$-th part to Player $i$. Given only the $i$-th part, Player $i$ then does not have incentive to change to any other strategy. Correlated equilibrium captures many natural scenarios that Nash equilibrium fails to do, as illustrated by the following two canonical examples. 

The first example is a game called \emph{Traffic Light}, in which two cars face each other at an intersection. Both cars have choices of passing and stopping. If both cars choose to pass, then there will be an accident, so both players suffers a lot. If at most one car passes, then the passing car has payoff 1 since it does not need to wait, and the car that stops has payoff 0. The payoff is summarized by the following payoff bimatrix, where in each entry, the first number is the payoff for Player 1 and the second is for Player 2.
\begin{center}
	\begin{tabular}{|r|c|c|}
		\hline 
		  & Cross & Stop \\
		\hline 
		Cross & (-100,-100) & (1,0) \\ 
		\hline
		Stop & (0,1) & (0,0)	\\
		\hline	
	\end{tabular}
\end{center}
There are two pure Nash equilibria in this game, namely (Cross,Stop) and (Stop,Cross). But there is a fairness issue: Which car should cross, given that both cars prefer so? Or in the language of games, which equilibrium they should agree on? There is actually a third Nash equilibrium, which is a mixed one: Each car crosses with probability 1/101. This solves the fairness issue, but lose the efficiency: The expected total payoff is very small (0);  most likely both cars would stop, and even worse, there is a positive probability of car crash. If one looks at the real world, things are much simpler by introducing a traffic light. Each car gets a signal which can be viewed as a random variable uniformly distributed on \{red, green\}. The two random signals/variables are designed to be perfectly correlated that if one is red, then the other is green. This is actually a correlated equilibrium, i.e. a distribution over \{Cross,Stop\}$\times$\{Cross,Stop\} with half probability on (Cross,Stop) and half on (Stop,Cross). It is easy to verify that it simultaneously achieves high payoff, fairness, and zero-probability of accident. 

The second example is a game called \emph{Battle of the Sexes}, in which a couple want to travel to a city for vacation, and \alice prefers A to B, while \bob prefers B to A. But both would like to visit the same city together rather than going to different ones separately. The payoffs are specified by the following bimatrix. \begin{center}
	\begin{tabular}{|r|c|c|}
		\hline 
		  & A & B \\
		\hline 
		A & (2,4) & (0,0) \\ 
		\hline
		B & (0,0) & (4,2)	\\
		\hline	
	\end{tabular}
\end{center}
Again, there are two pure Nash equilibria and the two parties prefer different ones, thus resulting a ``Battle" of the Sexes. A good solution is to take the correlated equilibrium, $(A,A)$ with half probability and $(B,B)$ with half probability, generated by a mediator flipping a fair coin. 

Apart from providing a natural solution concept in game theory as illustrated above, correlated equilibria also enjoy computational amenity for finding and learning in general strategic games as well as other settings such as graphical games (\hspace{-.08em}\cite{VNRT07}, Chapter 4 and 7).

\subsection{Quantum games}
Since there is no reason to assume that people interacting with quantum information are not selfish, quantum games provide a ground for understanding, reasoning and governing quantum interactions of selfish players, and it is thus important to investigate quantum games. The existing literature under the name of ``quantum games" can be roughly divided into three tracks. 
\begin{enumerate}
	\item \emph{Nonlocal games.} This is a particular class of Bayesian games in the strategic form, such as GHZ game, CHSH game, Magic Square game, etc. These games are motivated by the non-locality of quantum mechanics as opposed to any classical theory depending on ``hidden variables". In these games, each of the two or more parties receives a private input drawn from some known distribution, and the players output some random variables, targeting a particular correlation between their outputs and inputs. The main goal of designing and studying these games is to show that some correlations are achievable by quantum entanglement but not classical randomness, thus providing more examples for Bell's theorem \cite{Bel65} that refutes Einstein's program of modeling quantum mechanics as a classical theory with hidden variables. See \cite{BCMdW10} for a more comprehensive survey (with an emphasis on connections to communication complexity). In recent years non-local games also found connections to multi-prover interactive proof systems in computational complexity theory; see, for example, \cite{CHTW04,KKM+08,IKP+08,KKMV09,KR10,KRT10}.
	
	\item \emph{Quantization of strategic games}. Unlike the first track of research motivated by physics (and computational complexity theory), the second track of work aims at quantizing classical strategic game theory. The basic setting for a classical strategic game of $k$ players is as follows. Player $i$ has a set $S_i$ of strategies and a utility function $u_i$; when the players take a joint strategy $s = (s_1, \ldots, s_k)$, namely Player $i$ takes strategy $s_i$, each Player $i$ gets a payoff of $u_i(s)$. There are various models proposed to quantize this classical model. The basic approach is to extend each Player $i$'s strategy space from $S_i$ to the Hilbert space $H_i = span(S_i)$, and to allow the player to take quantum operations on $H_i$. Eventually a measurement in the computational basis is made to get a (random) classical joint strategy $s$, which decides the payoff of the players by the classical payoff functions $u_i$. %Variation exists in details of different proposed models. 
	
	The approach was implemented in the seminal paper \cite{EWL99} as follows; see Fig \ref{fig: EWL}. There is an extra party, called Referee, who applies a unitary operation $J$ on $\ket{0}$ (in the Hilbert space of dimension $\sum_i |S_i|$), and partitions the state $J\ket{0}$ into $k$ parts for the $k$ players. The players then perform their individual quantum operations on their own spaces, after which Referee collects these parts,  performs the inverse operation $J^{-1}$, and finally measures the state in the computational basis to get a random joint strategy $s$. Players $i$ then gets payoff $u_i(s)$.  
	
\begin{figure}%
\begin{center}
\includegraphics[width=4in]{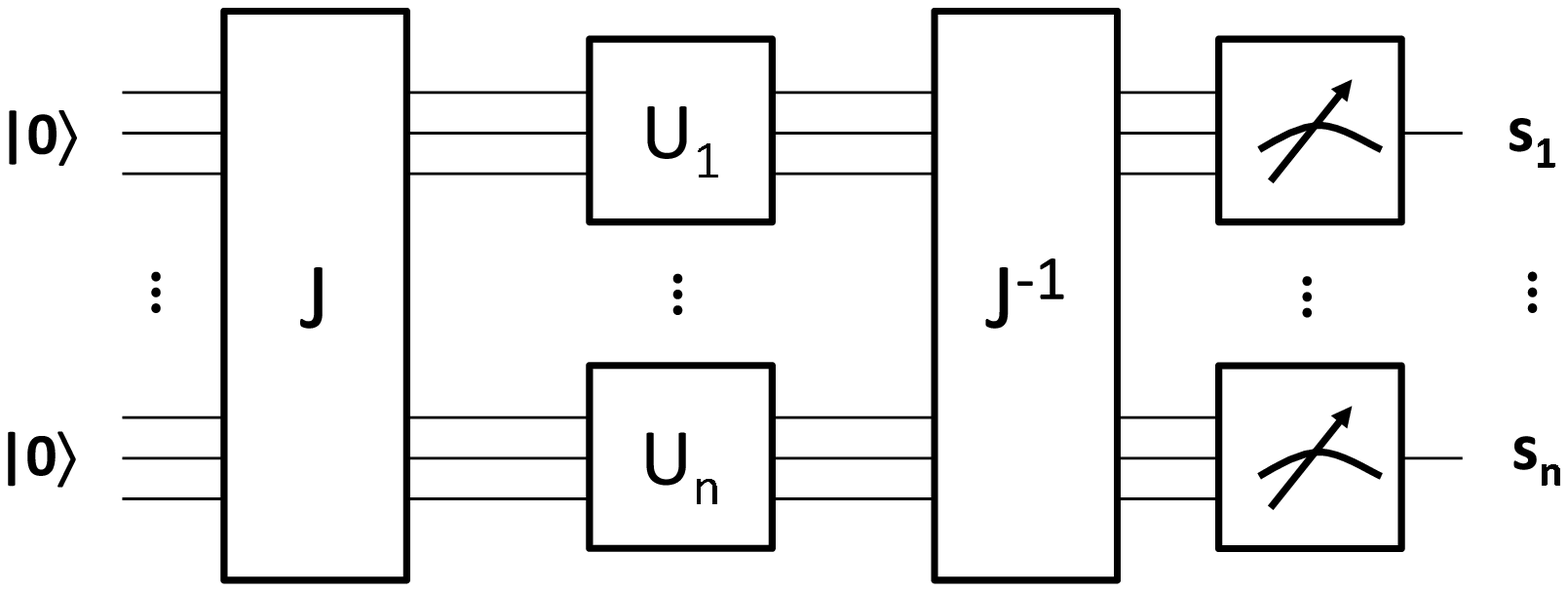}%
\end{center}
\caption{The EWL model for quantization of strategic games}%
\label{fig: EWL}%
\end{figure}

The EWL-model \cite{EWL99} unleashed a sequence of following studies under the same model \cite{BH01b,LJ03,FA03,FA05,DLX+02b,DLX+02a,PSWZ07}. Despite the rapid accumulation of literature on the same or similar model, controversy also exists. As pointed out in \cite{CT06}, there are ``ad hoc assumptions and arbitrary procedures scattered in the field". We will elaborate on this shortly.

	\item \emph{Quantum extensive games.} In a seminal work \cite{Mey99}, Meyer showed that in the classical Penny Matching game, if (1) Player 1 is allowed to use quantum strategies but Player 2 is restricted to classical strategies, and (2) the sequence of moves is (Player 1, Player 2, Player 1), then Player 1 can win the game for sure. This demonstrates the power of using quantum strategies under some particular restriction on the other player's strategies as well as the sequence of moves. Gutoski and Watrous \cite{GW07} initializes studies of the general refereed game in the extensive form. The model adopted there is very general, easily encompassing all previous work (and the model in our paper) as special cases. It has interesting applications such as a very short and elegant proof of Kitaev's lower bound for strong coin-flipping. The generality makes the framework and techniques potentially useful in a broad range of applications, though probably also admits less structures or at least makes it challenging to discover strong properties. Other examples of quantum extensive games include \cite{JW09,GW10}, which usually have a very small number of rounds. 
\end{enumerate}

\subsection{Our Results} 
Our goal is to study quantitative problems of general strategic games of size $n$ in a natural quantization model. To this end, we first give an arguably more natural model, and then study two measures of quantum advantages. 
\subsubsection{Model}
Despite of the prevalence, controversy also exists on the EWL-model. The main result in \cite{EWL99} was that a quantum strategy can ``escape" the Prisoner's dilemma, and this was obtained on the assumption that each player is only allowed to apply a specific subset of unitary operations. As pointed out in \cite{BH01}, the assumption does not seem to ``reflect any physical constraint (limited experimental resources, say) because this set is not closed under composition". Also shown in the paper \cite{BH01} is that without the assumption, namely if the players are allowed to use arbitrary local unitary operations, the proposed strategy in \cite{EWL99} is not a quantum Nash equilibrium any more. For this reason, we do not want to restrict players' possible actions in any way; we allow each player to take any quantum admissible operation (i.e. any TPCP map).

\begin{figure}%
\begin{center}
\includegraphics[width=2.5in]{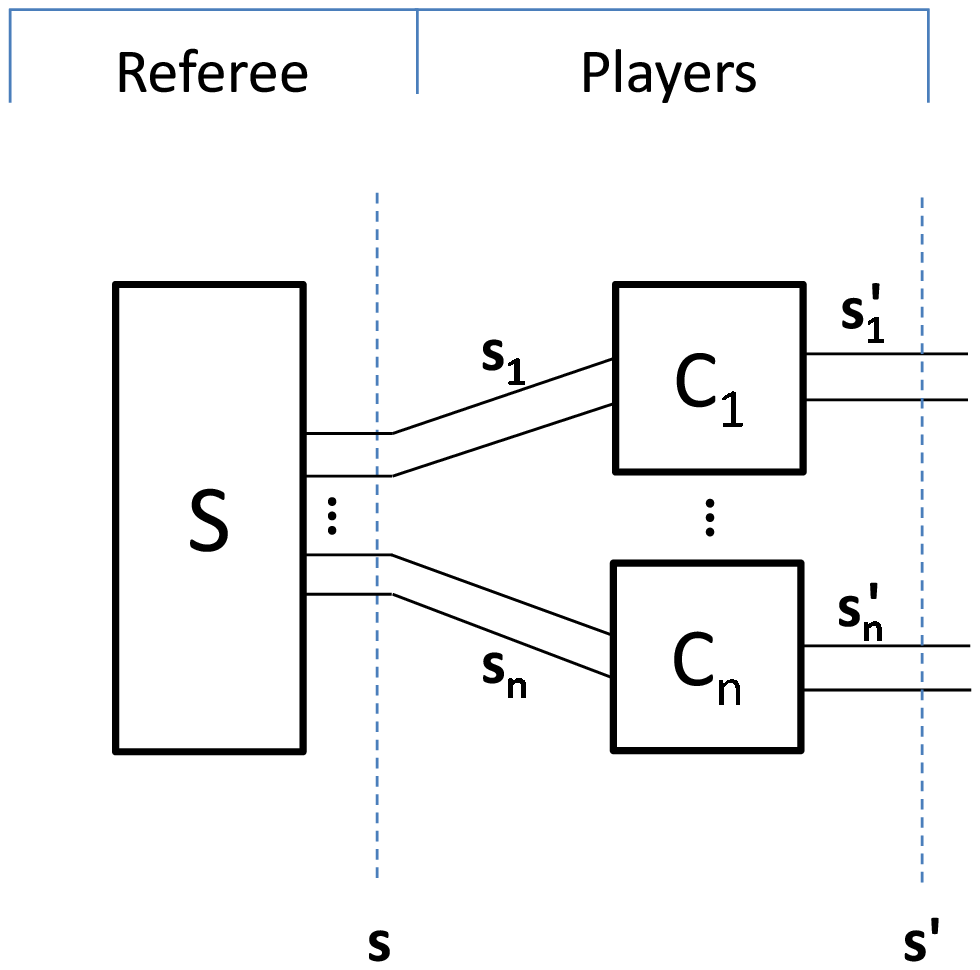}%
\end{center}
\caption{Classical strategic games: Referee samples a joint strategy $s$ and send the $i$-th part to Player $i$, who then applies a classical operation $C_i$ resulting in a possibly different strategy $s_i'$.}%
\label{fig: classical model}%
\end{figure}

A bigger difference of the EWL-model and ours, illustrated in Figure \ref{fig: our model}, is that we remove operation $J^{-1}$ in the EWL-model. We find that this corresponds to the classical model more precisely. Recall that in a classical strategic game, illustrated in Figure \ref{fig: classical model}, Referee samples a joint strategy $s = (s_1, \ldots, s_k) \in S$ from a classical distribution $p$ on $S$, and gives $s_i$ to Player $i$, who may apply a classical operator $C_i$ and output a possibly different strategy $s_i'$. The players then receive a payoff $u_i(s_1', ..., s_k')$. Note that different than in the EWL-model, Referee in the classical model does \emph{not} undo the initial sampling. 

A related question is why not going to the more general setting by letting Referee apply another joint operation $K$ before the final measurement?\footnote{The same question in another form: Why not allow a general measurement instead of the measurement in the computational basis?} Because classically Referee does not do any joint re-sampling after players' actions as well --- Referee's role is simply to sample and \emph{recommend} strategies to players. %Recall the Traffic Light game: Why does the traffic light want to mess up the signals at all? %(In particular, Referee does not perform any operation across different players' spaces.) 
Another advantage of not having $K$ is that now fundamental concepts such as quantum equilibrium (that we shall define next) will be only of the classical game under quantization, rather than also of an extra introduced quantum operator $K$. Last, if one really prefers to have $K$, then which $K$ to choose? In many games such as the two canonical examples in Section \ref{sec: intro}, Nature gives the payoff and Nature does not perform any joint measurement. (Consider for example the Traffic Light game: After the two cars get the signals and decide their moves, they do not send their pass/stop decision to any Referee for any joint measurement --- They simply perform the actions and then naturally face the consequences.) %there is no referee perform any operation across different players' spaces after they make their choices of strategies. 
So even if one likes to study various $K$'s, the case of $K=I$ should be probably the first natural one to consider. 

A final remark about the generality of the model: It is admitted that there are many ways to further generalize our model (such as having the general measurement $K$ discussed just now). But models should not be simply measured by generality, otherwise Nash equilibrium should not have been separately studied because it has so many (natural!) generalizations, and strategic games should not have been separately studied because they are just a special case of extensive games (two-move imperfect extensive games). Our goal was never to identify the most general model (which probably does not exist at all), but to propose a model which is natural, simple, fundamental, and hopefully rich in interesting questions --- like the notion of Nash equilibrium or the model of classical strategic games. 

\begin{figure}%
\begin{center}
\includegraphics[width=3.3in]{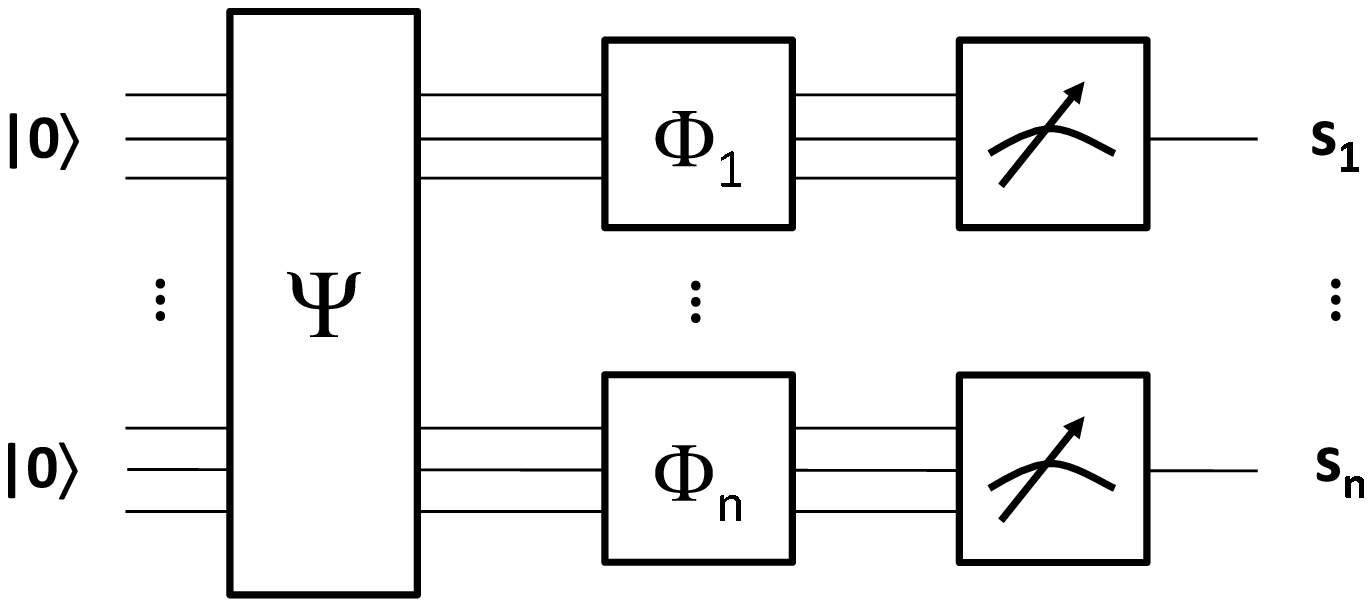}%
\end{center}
\caption{Our model for quantization of strategic games: No action of Referee after the players' moves, and the operations by Referee and the players are general quantum admissible ones.}%
\label{fig: our model}%
\end{figure}

So our model finally looks like the one in Figure \ref{fig: our model}: Referee applies a joint operator $\Psi$ on a all-zero state to create a quantum state $\rho$, and gives the $i$-th part of it to Player $i$, who applies $\Phi_i$ followed by a measurement in the computational basis. The players then receive their payoffs according to the functions $u_i$.  

Without the referee's action $J^{-1}$, our model is simpler. %The first question is whether such a simple model is still of game-theoretical interest. 
In \cite{BH01b}, three criteria were raised for an ideal quantization of classical strategic games given : (a) $S_i$ is generalized to $H_i = span(S_i)$, (b) strategies in $H$ are to be entangled, and (c) the resulting game generalizes the classical game. Note that despite being simpler than the EWL model, ours easily satisfies all of them as well. One may wonder whether ours is too simple to be of any mathematical interest. It turns out, as will be shown in the following sections, that our model has many interesting mathematical questions with connections to communication complexity, non-convex optimization and linear algebra. 

The concept of equilibria can be naturally extended to the quantum case. Recall that in a classical game, a joint strategy $s = (s_1, ... , s_k) \in S$ is sampled from a classical distribution $p$ on $S$ and Player $i$ receives an expected payoff $\av_{s\leftarrow p}[u_i(s)]$. A classical distribution $p$ is an equilibrium if no player can increase her expected payoff by any classical local operation. Now our model admits an almost word-by-word translation of the above definition to the quantum case: A joint strategy $s = (s_1, \ldots, s_k) \in S$ is measured from a quantum mixed state $\rho$ on $H$, and Player $i$ receives an expected payoff $\av_{s\leftarrow \rho}[u_i(s)]$. A quantum state $\rho$ is an equilibrium if no player can increase her expected payoff by any quantum local operation. Here the measurement is in the computational basis $S$, only on which the utility function is defined in the first place. 

%Different than the one in \cite{EWL99} which only allows a proper subset of unitary operations, we allow arbitrary CPTP maps for each player in her local space $H_i = span(S_i)$. In addition, no referee is involved; namely the operations $J$ and $J^{-1}$ are removed. We find this corresponds to the classical model more precisely because, after all, a classical strategic game does not have a classical referee to shuffle the whole strategy space $S = \times S_i$ before and after the players choose their strategies. 

%One may argue that we should allow a referee to do any joint measurement for the generality. First, in many natural games, Nature is the referee and Nature does not give any joint measurement. Think of the game... Second, classically why not do a joint re-sampling? Third, an important special case.

\subsubsection{Question}
 	Other than the model, what also distinguishes the present work from previous ones is the generality of the \emph{classical} games under quantization. Most of the previous work focus on particular games, usually of small and fixed sizes. For example, \cite{EWL99,DLX+02b,DLX+02a,PSWZ07,CH06} considered the \emph{Prisoner's Dilemma} game, \cite{MW00} considered the \emph{Battle of the Sexes} game, and \cite{Mey99} considered the \emph{Penny Matching} game, and there are many other studies on specific $2\times 2$ or $3\times 3$ games, e.g. \cite{FA03,FA05,DKK+02,IT02,DKK+02}, just to name a few. 
 	
 	In addition, most of the previous work focused on \emph{qualitative} questions such as whether playing quantum strategies has any advantage. While it is natural to start at qualitative questions on specific and small examples, it is surely desirable to have a systematic study on quantitative properties for general games. In particular, our aim is to understand the following fundamental question on general games of size $n$.
 	
 	\begin{quote}
 		\emph{Central Question: How much advantage can playing quantum strategies provide, if any?}
 	\end{quote}
	
Depending on how the advantage is measured, we study the question in two ways, summarized as follows.
\subsubsection{Quantum advantage 1: Increase of payoff}
	 Since games are all about players trying to get maximum payoffs, the first measure (of advantage) we naturally take is the increase of payoffs. We shall consider natural mappings between classical and quantum states, and study how well those mappings preserve the equilibrium properties. Recall that a quantum state $\rho$ in space $H = \otimes_i H_i$ is a quantum correlated equilibrium if no Player $i$ can increase her expected payoff by any local operation. If further $\rho = \otimes_i \rho_i$ for some $\rho_i$ in $H_i$, then it is a quantum Nash equilibrium. %The definition encompasses the classical correlated and Nash equilibria as special cases. 
	 
	Under this definition, we relate classical and quantum equilibria in the following ways. Given a quantum state, the most natural classical distribution it induces is given by the measurement in the computational basis $S$. That is, $\rho$ induces $p$ where $p(s) = \rho_{s,s}$. Not surprisingly, one can show that if $\rho$ is a quantum Nash (or correlated) equilibrium then $p$ is a classical Nash (or correlated) equilibrium.
	
	The other direction, namely transition from classical to quantum, is more complicated but interesting. A classical distribution $p$ over $S$ has two natural quantum counterparts: 1) \emph{classical mixture}: $\rho(p) = \sum_s p(s) \ket{s}\bra{s}$, the mixture of the classical states, and 2) \emph{quantum superposition}: $\ket{\psi(p)} = \sum_s \sqrt{p(s)} \ket{s}$. We regard the second mapping as more important because firstly, this is really quantum --- the first mapping is essentially the classical state itself --- and secondly, this mapping is the most commonly used quantum superposition of a classical distribution in known quantum algorithms, such as starting state in Grover's search \cite{Gro97} and the states to define the reflection subspaces in Szegedy's quantization of random walks \cite{Sze04}. It so happens that it is also the most intriguing case of our later theorems. 
	
	One can also consider the broad class of quantum states $\rho$ satisfying $p(s) = \rho_{s,s}$, including the above two concrete mappings as special cases. Now the question is, do these transformations keep the Nash/correlated equilibrium properties? It turns out that the classical mixture mapping keeps both Nash and correlated equilibrium properties, but the quantum superposition mapping only keeps the Nash equilibrium property. As to the general class of correspondence, no equilibrium is guaranteed to be kept. 

	Based on these answers, it is more desirable to study them quantitatively: After all, if $\ket{\psi(p)}$ is not an exact correlated equilibrium but always an $\epsilon$-approximate one, in the sense that no player can increase her payoff by more than a small amount $\epsilon$, then the interest of using quantum strategies significantly drops. Therefore, we are facing the following question. (For proper comparison, assume that all games are [0,1]-normalized, \ie all utilities take values from [0,1].)
	\begin{quote}
		\emph{Question 1: In a [0,1]-normalized game, what is the largest gain of payoff by playing a quantum strategy on a quantum counterpart state of a classical equilibrium}?
	\end{quote} 
The question turns out to be a non-convex program, which is notoriously hard to analyze in general. Actually even the simple case of $n=2$ is already quite nontrivial to solve. %(For FOCS: For a warm-up, we recommend readers try to come up with a game and a CE $p$ that playing quantum strategies on $\ket{\psi_p}$ can have a positive gain.) 
The maximum gain turns out to be a small constant close to 0.2, but neither the analysis nor the solution admits a generalization to higher dimensions in any straightforward way. For general $n$, there is no clue what the largest gain should be. Nevertheless, we could show the following, among other results. 	
	\begin{Thm}\label{thm:main1}
	\begin{enumerate}
		\item There exists a correlated equilibrium $p$ in a $[0,1]$-normalized $(n\times n)$-bimatrix game \st 
		\begin{equation}
			u_1(\ket{\psi(p)}) = \tilde O(1/\log n) \quad \text{ and } \quad u_1(\Phi_1(\ket{\psi(p)})) = 1-\tilde O(1/\log n),
		\end{equation}  
		for some local quantum operation $\Phi_1$.\footnote{$\tilde O$ hides a $poly(\log\log(n))$ factor.} There is also a correlated equilibrium $p$ with the multiplicative factor 
		\begin{equation}
			\frac{u_1(\Phi_1(\ket{\psi(p)}))}{u_1(\ket{\psi(p)})} = n^{0.585...}.
		\end{equation}

		\item There exists a Nash equilibrium $p$ in a $[0,1]$-normalized $(n\times n)$-bimatrix game, and a quantum state $\rho$ with $\rho_{ss} = p(s)$, \st 
		\begin{equation}
			u_1(\rho) = 1/n \quad \text{ and } \quad u_1(\Phi_1(\rho)) = 1,
		\end{equation} 
		for some local quantum operation $\Phi_1$. The additive increase of $1-1/n$ and the multiplicative increase of $n$ are the largest possible even for all \emph{correlated} equilibria $p$.
	\end{enumerate}
 	\end{Thm}
Note that optimality is proved in the second part, and the upper bounds of the maximum gain apply to $\ket{\psi(p)}$ in the first part as a special case. Closing the gaps between the lower bounds in the first part and the general upper bounds in the second part is left open. 
 	
The main approach for Part 1 is to construct large games from smaller ones. What we need for the construction is to preserve the equilibrium \emph{and} to increase the ``quantum gain", the gain by playing quantum strategies. It turns out that the tensor product preserves the equilibrium property, and can increase the gain for small games with some parameters. The design of the base games is also not straightforward: Taking the optimal solution to the $n=2$ case does not work because taking power on that game actually decreases the gain. In the final solution, the base game itself has a very small quantum gain, but when taken power, the classical-strategy utility drops much faster than the quantum-strategy utility, creating a gap almost as large as 1. 

\subsubsection{Quantum advantage 2: Correlation generation}
	We also study the quantum advantage from a complexity-theoretical perspective. As we have mentioned, correlated equilibria possess game theoretical usefulness and enjoy better computational tractability. But to really use such a good equilibrium, someone has to \emph{generate} it, which makes the hardness of its generation an interesting question. For this, we propose a new complexity measure, called \emph{correlation complexity}, defined as follows. 
	
	Take two-party case, for simplicity, where \alice and \bob aim to generate a correlation. Since local operation cannot create correlation, they start from some ``seed", which can be either a shared classical randomness or a quantum entangled state. Then they perform local operations and finally output the target correlation. %(In the original setting of the correlated equilibria generation, the referee generates the seed and send it to the two players. Then the communication complexity is the same as our correlation complexity.) 
	We are concerned with the following question.
	
	\begin{quote}
		\emph{Question 2: To generate the same correlation, does quantum entanglement as a seed have any advantage compared to the classical shared randomness? If yes, how much?}
	\end{quote}
	
	Note that this question is, in spirit, not new. Actually the entire class of non-local games study questions of the same flavor. However, a crucial part in non-local games is that the two parties are given \emph{private} (and random) inputs, which are necessary for differentiating the power of classical hidden variable and that of quantum entanglement in previous non-local game results. 

	Without the private inputs, our model is simpler and thus more basic. An immediate question is whether such a bare model still admits any separation of classical and quantum powers in generating correlations. This paper gives a strongly affirmative answer.

\begin{Thm}\label{thm:corr separation}
	For any $n>2$, there are correlations $(X, Y)$ which take at least $n$ classical bits to generate classically, but only need \emph{one} EPR pair to generate quantum mechanically.
\end{Thm}

In proving the classical lower bound, we identify the \emph{nonnegative rank} as the correct measure to fully characterize the randomized correlation complexity. The nonnegative rank is a well-studied measure in linear algebra and it has many applications to statistics, combinatorial optimization \cite{Yan88}, nondeterministic communication complexity \cite{Lov90}, algebraic complexity theory \cite{Nis91}, and many other fields \cite{CP05}. 

The hidden asymptotic lower bound for randomized correlation complexity of a size-$n$ correlation is actually $\Omega(\log n)$. The bound can be improved to $n$, the largest possible, assuming a recent conjecture in linear algebra \cite{BL09}. We actually have a bold conjecture that, with probability 1, a random correlation that can be generated by one EPR pair has the randomized correlation complexity of $n$. Note that $n$ always suffices since for any fixed correlation $(X,Y)$, the two parties can simply share this very same correlation as the seed and output it. So ``1 \emph{vs.} $n$" is the largest possibly separation; this is in contrast to Bell's inequality that even infinite amount of classical shared randomness cannot simulate one EPR pair. In this sense, the correlation complexity can be viewed as a sublinear complexity-theoretical counterpart of previous non-local games. 

Coming back to the setting of games, two scenarios can happen depending on whether the local operations are trusted or not. In the first scenario, consider the a generalized \emph{Battle of the Sexes} game, where \alice and \bob are not in the same city but want to generate some correlation $p=(X,Y)$. There is a publicly trusted company C, which can help to generate $p$. Company C has a central server which generates a seed and send to its local servers A and B, distributed close to \alice and \bob, respectively. The local servers A and B apply the local operations to generate a state which is then sent to \alice and \bob. Here the local operations are carried out by the trusted servers A and B. And the complexity that we care is the size of the seed, which is also the communication between the central server to the two distributed servers A and B. The separation of classical and quantum correlation complexities directly applies to this scenario. 

In the second scenario, the mediator sends the seed directly to \alice and \bob, who are then supposed to apply the local operations $\Phi_1$ and $\Phi_2$ to generate the CE $(X,Y)$. But since now the local operations are under the control of the players, they can apply some other local operations $\Phi_1'$ and $\Phi_2'$. So the process is an equilibrium if no player has an incentive to apply any other local operation. The above separation can still be adapted to separate the minimum sizes of classical and quantum seeds in some games, but in general this scenario is more complicated and less understood, leaving a good direction for future exploration.

%The correlation complexity is also connected to communication complexity. Actually, since seed correlation can be simulated by one message, the (even one-way) communication complexity is always a lower bound of the correlation complexity. Our separation, however, is robust to involve communication as well; namely, the quantum upper bound of 1 does not need any communication, but the randomized lower bounds of $\Omega(\log n)$ holds even for communication complexity. 

\subsection{More related work} 
	The last decade has witnessed the advance of our understandings of the hardness to find a Nash equilibrium in strategic games \cite{DGP09,CDT09}. There has also been some studies for communication complexity of finding a Nash equilibrium \cite{CS04, HM10}, when each player only knows her own utility function. 

	The problem of correlation generation in the asymptotic setting is considered in \cite{Wyn75} for the classical case and \cite{Win05} for the quantum case. The paper \cite{HJMR09} also studies the communication complexity for generating a correlation $(X,Y)$. But the model there takes an average-case measure: Suppose \alice samples $x\leftarrow X$ and tries to let \bob sample from $Y|(X=x)$, then what is the \emph{expected} communication needed (where the expectation is over the randomness of protocol as well as the initial sample $x\leftarrow X$)? For comparison, ours is a worst-case measure requiring that for each possible $x$, Bob samples from $Y|(X=x)$. And also note the essential difference that protocols in \cite{HJMR09} uses a large amount of public coins, which is exactly the resource we hope to save. See the last section for more discussions on this.
	
	After an earlier version of the present paper was finished and circulated, Yaoyun Shi firstly pointed out the paper \cite{ASTS+03}, which studies communication complexity of correlation generation. The correlations studied there, however, are a particular type, arising from communication complexity of Boolean functions, while ours considers general correlations. The second difference is that \cite{ASTS+03} only considers the \emph{communication} complexity, but ours also considers correlation complexity, the minimum shared resource (public randomness or entanglement) for generating the correlation \emph{without} any communication. It turns out that in the trusted local operation setting, correlation complexity is the same as communication complexity, both classically and quantumly. In the randomized case, we characterize them by nonnegative rank. The measures in the untrusted local operation setting are of a totally different story: While correlation complexity is still sublinear, there may not even be any equilibrium communication protocol to generate the correlation. Last, the main body in \cite{ASTS+03} studies a bounded-error generation, and showed an exponential separation ($O(\log n)$ versus $\Omega(\sqrt{n})$), while ours aims to generate the exact target correlation, and showed an ``infinite" separation (1 versus $\log_2 n$ unconditionally, and 1 versus $n$ assuming a conjecture).  

	Studies of computational issues of probabilistic distributions instead of Boolean functions has recently be advocated by Viola \cite{Vio10, LV10}. It is our hope that studies of the correlation complexity of distributions later help to sharpen our understandings of various complexity questions for Boolean functions.

\subsection*{Organization}
The rest of paper is organized as follows. In Section \ref{sec:pre}, after reviewing model for classical strategic games and the definitions of Nash and correlated equilibria, we introduce the quantum model and define quantum Nash and correlated equilibria. Other notation is also set up in the section. In Section 3, we show how natural maps between classical and quantum states preserves equilibrium properties, giving the proof of Theorem \ref{thm:main1}. Section 4 is devoted to the correlation complexity, where we show proof of Theorem \ref{thm:corr separation}. In the last section, we point out quite a number of problems and directions for future research.

\section{Preliminaries, quantum model, and notation}\label{sec:pre}
Suppose $X$ and $Y$ are two (possibly correlated) random variables on sample spaces $\mcX$ and $\mcY$, respectively. The \emph{size} of bivariate distribution $p = (X,Y)$, denoted by $\size(p)$, is defined as $(\lceil \log_2(|\mcX|)\rceil + \lceil \log_2(|\mcY|)\rceil ) / 2$. Here we take the factor of half because we shall talk about a correlation as a \emph{shared} resource. It is consistent with the convention that when $Y = X = R$, we say that they share a random variable $R$ of size $\lceil \log_2(|\mcX|)\rceil$. For a two-party quantum state $\rho$ in $H^1\otimes H^2$ for Hilbert spaces $H^i$ of dimension $D_i$, we also say that the size of the $\rho$, as a shared quantum state, is $(\lceil \log_2(D_1)\rceil + \lceil \log_2(D_2)\rceil ) / 2$.

Sometimes we view a bivariate distribution $p$ as a matrix, denoted by the capital $P$ for emphasis, where the row space is identified with $\mcX$ and the column space with $\mcY$. 

A matrix $A$ is called \emph{nonnegative} if each entry is a nonnegative real number. For a nonnegative matrix $A$, its \emph{nonnegative rank}, denoted by $\rank_+(A)$, is the minimum number $r$ such that $A$ can be decomposed as the summation of $r$ nonnegative matrices of rank 1. 

Suppose that in a classical game there are $k$ players, labeled by $\{1,2,\ldots,k\}$. Each player $i$ has a set $S_i$ of strategies. %To play the game, each player $i$ selects a strategy $s_i$ from $S_i$.
We use $s=(s_1,\ldots, s_k)$ to denote the \emph{joint strategy} selected by the players and $S= S_1 \times \ldots \times S_k$ to denote the set of all possible joint strategies. Each player $i$ has a utility function $u_i: S \rightarrow \mbR$, specifying the \emph{payoff} or \emph{utility} $u_i(s)$ to player $i$ on the joint strategy $s$. For simplicity of notation, we use subscript $-i$ to denote the set $[k]-\{i\}$, so $s_{-i}$ is $(s_1, \ldots, s_{i-1}, s_{i+1}, \ldots, s_k)$, and similarly for $S_{-i}$, $p_{-i}$, etc. 

A game is \emph{$[0,1]$-normalized}, or simply \emph{normalized}, if all utility functions have the ranges in $[0,1]$.

\subsection{Classical equilibria}
Nash equilibrium is a fundamental solution concept in game theory. Roughly, it says that in a joint strategy, no player can gain more by changing her strategy, provided that all other players keep their current strategies unchanged. The precise definition is as follows.
\begin{Def}
A \emph{pure Nash equilibrium} is a joint strategy $s = (s_1, \ldots ,s_k) \in S$ satisfying that
\begin{align*}
	u_i(s_i,s_{-i}) \geq  u_i(s_i',s_{-i}), \qquad \forall i\in [k], \forall s'_i\in S_i.
\end{align*}
\end{Def}

Pure Nash equilibria can be generalized by allowing each player to independently select her strategy according to some probability distribution, leading to the following concept of \emph{mixed Nash equilibrium}. 

\begin{Def}
A \emph{(mixed) Nash equilibrium (NE)} is a product probability distribution $p = p_1 \times \ldots \times p_k$, where each $p_i$ is a probability distributions over $S_i$, satisfying that
\begin{align*}
	\sum_{s_{-i}} p_{-i}(s_{-i}) u_i(s_i,s_{-i}) \geq  \sum_{s_{-i}} p_{-i}(s_{-i}) u_i(s_i',s_{-i}), \quad \forall i\in [k], \ \forall s_i, s'_i\in S_i \text{ with } p_i(s_i)>0.
\end{align*}
\end{Def}

Informally speaking, for a mixed Nash equilibrium, the expected payoff over probability distribution of $s_{-i}$ is maximized, i.e. $\av_{s_{-i}}[u_i(s_i,s_{-i})] \ge \av_{s_{-i}}[u_i(s_i',s_{-i})]$. A fundamental fact is the following existence theorem proved by Nash.

\begin{Thm}[Nash, \cite{Nas51}]\label{thm: existence of NE}
	Every game with a finite number of players and a finite set of strategies for each player has at least one mixed Nash equilibrium. 
\end{Thm}

There are various further extensions of mixed Nash equilibria. Aumann \cite{Aum74} introduced a relaxation called \emph{correlated equilibrium}. This notion assumes an external party, called Referee, to draw a joint strategy $s = (s_1, ..., s_k)$ from some probability distribution $p$ over $S$, possibly correlated in an arbitrary way, and to suggest $s_i$ to Player $i$. Note that Player $i$ only sees $s_i$, thus the rest strategy $s_{-i}$ is a random variable over $S_{-i}$ distributed according to the conditional distribution $p|_{s_i}$, the distribution $p$ conditioned on the $i$-th part being $s_i$. Now $p$ is a correlated equilibrium if any Player $i$, upon receiving a suggested strategy $s_i$, has no incentive to change her strategy to a different $s_i' \in S_i$, assuming that all other players stick to their received suggestion $s_{-i}$. 

\begin{Def} \label{thm:CE}
A \emph{correlated equilibrium (CE)} is a probability distribution $p$ over $S$ satisfying that
\begin{align*}
	\sum_{s_{-i}} p(s_i,s_{-i}) u_i(s_i,s_{-i}) \geq  \sum_{s_{-i}} p(s_i,s_{-i}) u_i(s_i',s_{-i}), \qquad \forall i\in [k], \ \forall s_i, s'_i\in S_i.
\end{align*}
\end{Def}
%The above statement can also be stated as $\av_{s_{-i} \gets \mu|s_i}[u_i(s_i,s_{-i})] \ge \av_{s_{-i} \gets \mu|s_i}[u_i(s_i',s_{-i})]$. 
Notice that a classical correlated equilibrium $p$ is a classical Nash equilibrium if $p$ is a product distribution. %For convenience, we sometimes just call a \emph{correlated equilibrium (CE)} for a classical correlated equilibrium, and a \emph{Nash equilibrium (NE)} for a classical Nash equilibrium.\\

Correlated equilibria captures natural games such as the Traffic Light and the Battle of the Sexes mentioned in Section 1. The set of CE also has good mathematical properties such as being convex (with Nash equilibria being some of the vertices of the polytope). Algorithmically, it is computationally benign for finding the best CE, measured by any linear function of payoffs, simply by solving a linear program (of polynomial size for games of constant players). A natural learning dynamics also leads to an approximate CE (\hspace{-.08em}\cite{VNRT07}, Chapter 4) which we will define next, and all CE in a graphical game with $n$ players and with $\log(n)$ degree can be found in polynomial time (\hspace{-.08em}\cite{VNRT07}, Chapter 7).

Another relaxation of equilibria changes the requirement of absolutely no gain (by deviating the strategy) to gaining a little, as the following approximate equilibrium defines.

%{\it // check whether Chen-Deng-Teng uses this worst-case def.}

%\begin{Def}
%An \emph{$\epsilon$-approximate Nash equilibrium ($\epsilon$-NE)} is a probability vector $p = p_1 \times \ldots \times p_n$ for some probability distributions $p_i$'s over $S_i$'s satisfying that
%\begin{align*}
%	\sum_{s_{-i}} p_{-i}(s_{-i}) u_i(s_i,s_{-i}) \geq  \sum_{s_{-i}} p_{-i}(s_{-i}) u_i(s_i',s_{-i})-\epsilon, \quad \forall i\in [n], \ \forall s_i, s'_i\in S_i \text{ with } p_i(s_i)>0.
%\end{align*}
%\end{Def}

%Analogous to approximate Nash equilibria, one can also define approximate correlated equilibria as follows. 
\begin{Def}\label{def: approx CE}
An \emph{$\epsilon$-additively approximate correlated equilibrium} is a probability distribution $p$ over $S$ satisfying that
\begin{align*}
	\av_{s\leftarrow p}[u_i(s_i'(s_i)s_{-i})] \leq \av_{s\leftarrow p}[u_i(s)] + \epsilon, 
\end{align*}
for any $i$ and any function $s_i': S_i \rightarrow S_i$. For such distributions $p$, we say that the maximum \emph{additive incentive} (to deviate) is the minimum $\epsilon$ with the above inequality satisfied. Furthermore, the distribution $p$ is called an \emph{$\epsilon$-additively approximate Nash equilibrium} if it is a product distribution $p_1 \times \ldots \times p_k$. 

An \emph{$m$-multiplicatively approximate correlated equilibrium} is a probability distribution $p$ over $S$ satisfying that
\begin{align*}
	\av_{s\leftarrow p}[u_i(s_i'(s_i)s_{-i})] \leq m \cdot \av_{s\leftarrow p}[u_i(s)], 
\end{align*}
for any $i$ and any function $s_i': S_i \rightarrow S_i$. For such distributions $p$, we say that the maximum \emph{multiplicative incentive} (to deviate) is the minimum $m$ with the above inequality satisfied. Furthermore, the distribution $p$ is called an \emph{$\epsilon$-multiplicatively approximate Nash equilibrium} if it is a product distribution $p_1 \times \ldots \times p_k$. 
\end{Def}
%First, it is easily seen that the concept coincides with the standard correlated equilibrium when $\epsilon = 0$. Second, 
Note that one can also define a stronger notion of approximation %one difference between this and the one can also define a stronger version of the approximate correlated equilibrium 
by requiring that the gain is at most $\epsilon$ for each possible $s_i$ in the support of $p$. Definition \ref{def: approx CE} only requires the gain be small on average (over $s_i$), but it is usually preferred because of its nice properties, such as the aforementioned result of being the limit of a natural dynamics of minimum regrets (\hspace{-.08em}\cite{VNRT07}, Chapter 4). 

\subsection{Quantum equilibria}
In this paper we consider quantum games which allows the players to use strategies quantum mechanically. We assume the basic background of quantum computing; see \cite{NC00} and \cite{Wat08} for comprehensive introductions. The set of admissible super operators, or equivalently the set of completely positive and trace preserving (CPTP) maps, of density matrices in Hilbert spaces $H_A$ to $H_B$, is denoted by $\cptp(H_A, H_B)$. We write $\cptp(H)$ for $\cptp(H,H)$.

For a classical strategic game as we have discussed so far, when being played quantumly, each player $i$ has a Hilbert space $H_i = span\{s_i: s_i\in S_i\}$, and a joint strategy can be any quantum state $\rho$ in $H = \otimes_i H_i$. Since we want to quantize classically defined games rather than creating new rules, we respect the utility functions of the original games. Thus we only talk about utility when we get a classical joint strategy. The most, if not only, natural way for this is to directly measure in the computational basis, which corresponds to the classical strategies. Therefore the (expected) payoff for player $i$ on joint strategy $\rho$ is 
\begin{equation}
	u_i(\rho) = \sum_s \bra{s} \rho \ket{s} u_i(s).
\end{equation}
In summary, the players measure the state $\rho$ in the computational basis $S$, resulting in a distribution of the joint strategies, and the utility is just the expected utility of this random joint strategy. 

Corresponding to changing strategies in a classical game, now each player $i$ can apply an arbitrary CPTP operation on $H_i$. So the natural requirement for a state being a quantum Nash equilibrium is that each player cannot gain by applying any admissible operation on her strategy space. The concepts of quantum Nash equilibrium, and quantum correlated equilibrium, and quantum approximate equilibrium are defined in the following, where we overload the notation by writing $\Phi_i$ for $\Phi_i \otimes I_{-i}$ if no confusion is caused. 

\begin{Def}
A \emph{quantum Nash equilibrium (QNE)} is a quantum strategy $\rho = \rho_1 \otimes \cdots \otimes \rho_k$ for some mixed states $\rho_i$'s on $H_i$'s satisfying that 
\begin{align*}
	u_i(\rho) \ge u_i(\Phi_i(\rho)), \qquad \forall i\in [k], \ \forall \Phi_i \in \cptp(H_i).
\end{align*}
\end{Def}

\begin{Def}
An \emph{$\epsilon$-approximate quantum Nash equilibrium ($\epsilon$-QNE)} is a quantum strategy $\rho = \rho_1 \otimes \ldots \otimes \rho_n$ for some mixed states $\rho_i$'s in $H_i$'s satisfying that 
\begin{align*}
	u_i(\Phi_i(\rho)) \leq u_i(\rho) + \epsilon, \qquad \forall i\in [k], \forall \Phi_i \in \cptp(H_i).
\end{align*}
\end{Def}

\begin{Def} \label{def:QCE}
A \emph{quantum correlated equilibrium (QCE)} is a quantum strategy $\rho$ in $H$ satisfying that 
\begin{align*}
	u_i(\rho) \ge u_i(\Phi_i(\rho)), \qquad \forall i\in [k], \ \forall \Phi_i \in \cptp(H_i).
\end{align*}
\end{Def}

\begin{Def}
An \emph{$\epsilon$-additively approximate quantum correlated equilibrium ($\epsilon$-QCE)} is a quantum state $\rho$ in $H$ satisfying that
\begin{align*}
	u_i(\Phi_i(\rho)) \leq u_i(\rho) + \epsilon, 
\end{align*}
for any $i$ and any admissible map $\Phi_i$ on $H_i$. For such states $\rho$, we say that the maximum \emph{quantum additive incentive} (to deviate) is the minimum $\epsilon$ with the above inequality satisfied. 

An \emph{$m$-multiplicatively approximate quantum correlated equilibrium ($\epsilon$-QCE)} of a nonnegative utility game is a quantum state $\rho$ in $H$ satisfying that
\begin{align*}
	u_i(\Phi_i(\rho)) \leq m\cdot u_i(\rho), 
\end{align*}
for any $i$ and any admissible map $\Phi_i$ on $H_i$. For such states $\rho$, we say that the maximum \emph{quantum multiplicative incentive} (to deviate) is the minimum $m$ with the above inequality satisfied. 
\end{Def}

One can also extend the $\epsilon$-QCE by allowing different $\epsilon_i$ for different $i$, resulting in $\{\epsilon_i\}$-QCE. By the linearity of admissible map $\Phi_i$, of quantum utility function $\mu_i$, and of expectation, it is easily seen that for any $\{\epsilon_i\}$, the set of $\{\epsilon_i\}$-QCE is convex. In particular, the set of QCE is also convex. Similar to the classical case, a quantum correlated equilibrium $\rho$ is a quantum Nash equilibrium if $\rho$ is a product state.

A final remark about QNE: One may wonder why not allow separable states, namely $\rho = \sum_t p_t (\rho_{t,1}\otimes \cdots \otimes \rho_{t,k})$ for some distribution $p$ and quantum states $\rho_{t,i}\in H_i$. The reason is that correlation then exists between players, so it includes the classical correlated equilibria as special cases. Our preference here is to let QNE to cover NE and QCE to cover CE, but QNE should not cover CE.

\section{Translations between classical and quantum equilibria}
This section studies the relation between classical and quantum equilibria. Basically we would like to consider all natural correspondences between classical and quantum states, and see how well they preserve the equilibrium properties. Thus there are two directions of mappings: from quantum to classical and and from classical to quantum. We will first list the correspondences and study them in detail in the subsections.

For the first direction, the most natural way to get a classical distribution from a quantum state is, as mentioned, to measure it in the computational basis:
\begin{equation}
	p(s) = \rho_{ss}, \text{ where } \rho_{ss} \text{ is the $(s,s)$-th entry of the matrix } \rho.
\end{equation}

Next we consider mappings from classical distributions $p$ over $S$ to quantum states on $H$. There seem to have more natural options. As far as we can think of, there are two specific mappings and a big class of correspondences including the two as special cases.
\begin{enumerate}
	\item \textbf{classical mixture}: $\rho(p) = \sum_s p(s) \ket{s}\bra{s}$, the mixture of the classical states. This is essentially an identity map, though when playing the quantum game the players are allowed to perform any quantum operations on it.
	\item \textbf{quantum superposition}: $\ket{\psi(p)} = \sum_s \sqrt{p(s)} \ket{s}$. With the superposition, this is really quantum and we expect to see some interesting and nontrivial phenomena. This is the most commonly used quantization of probability distributions when designing quantum algorithms. For example, recall that the starting state of Grover's search \cite{Gro97} and the states to define the reflection subspaces in Szegedy's quantization of random walks \cite{Sze04} are both of this form.  
	\item \textbf{general correspondence}: any density matrix $\rho$ with $p(s) = \rho_{ss}$ satisfied for all $s\in S$. This is the least requirement we want to put, and it is a large set of mappings containing the first two as special cases.
\end{enumerate}

Next we address the questions whether being equilibria in one world, classical or quantum, implies equilibria in the other world, and if not, how bad it can be. % Some immediate questions can be asked: If $\rho$ is a quantum NE/CE, what can we say about $p_\rho$? In the other direction, what can we say about $\rho(p)$, $\ket{\psi(p)}$, or any density matrix $\rho$ with $p(s) = \rho_{ss}$ satisfied, for a classical NE/CE $p$? This section is devoted to addressing these questions. 

\subsection{From quantum to classical}
The following theorem says that the quantum equilibrium property always implies the classical one. The proof is not hard; one catch is that what we know for $\rho$ is that any quantum operation on $H_i$ cannot increase the \emph{expected} payoff. What we need to prove is, however, a \emph{worst-case} statement, namely that for any Player $i$ and any received strategy $s_i$, she should not change to any other $s_i'$. We just need to handle this distinction.

\begin{Thm}[QCE $\Rightarrow$ CE, QNE $\Rightarrow$ NE]\label{thm: QE to CE}
	If $\rho$ is a quantum correlated equilibrium, then $p$ defined by $p(s) = \rho_{ss}$ is a classical correlated equilibrium. In particular, if $\rho$ is a quantum Nash equilibrium, then $p$ is a classical Nash equilibrium.
\end{Thm}

%Before presenting the formal proof, we would like to discuss the intuition first which is actually pretty clear: What we know for $\rho$ is that any quantum operation on $H_i$ cannot increase the \emph{expected} payoff. What we need to prove is, however, a \emph{worst-case} statement, namely that for any Player $i$ and any received strategy $s_i$, she should not change to any other $s_i'$. This distinction crucially affects some of our results (such as the later Theorem \ref{}), but fortunately not here: One simply needs to pick a quantum operation to make the average increase of payoff equal to the that for any fixed worst-case classical changing of strategies. A formal proof is as follows.

\begin{proof}
Recall that we are given that $\mu _i(\rho) \ge \mu _i(\Phi_i(\rho) )$ for all players $i$ and all admissible super-operators $\Phi_i$ on $H_i$, and we want to prove that for all players $i$ and all strategies $s_i, s'_i\in S_i$,
\begin{equation}\label{eq: CE}
	\sum_{s_{-i}} p(s_i,s_{-i}) u_i(s_i,s_{-i}) \geq  \sum_{s_{-i}} p(s_i,s_{-i}) u_i(s_i',s_{-i})
\end{equation}
for $p(s) = \rho_{ss}$.

Fix $i$ and $s_i, s'_i$. Consider the admissible super-operator $\Phi_i$ defined by
\begin{equation}
	\Phi_i = \sum_{t_i \neq s_i} P_{t_i} \rho P_{t_i} + (s_i \leftrightarrow s_i') P_{s_i} \rho P_{s_i} (s_i \leftrightarrow s_i')
\end{equation}
where $P_{t_i}$ is the projection onto the subspace $span(t_i)\otimes H_{-i}$, and $(s_i \leftrightarrow s_i')$ is the operator swapping $s_i$ and $s_i'$. It is not hard to verify that $\Phi_i$ is an admissible super-operator. Next we will show that the difference of $\mu_i(\rho)$ and $\mu_i(\Phi_i(\rho))$ is the same as that of the two sides of Eq. \eqref{eq: CE}.

\begin{align}
\nonumber \mu_i(\rho) & = \av[u_i(s(\rho))]\\
\nonumber & = \sum_{\bar s\in S} \bra{\bar s}\rho \ket{\bar s} u_i(\bar s) = \sum_{\bar s\in S} p(\bar s) u_i(\bar s) \\
 & = \sum_{\bar s_i \neq s_i} \sum_{\bar s_{-i}} p(\bar s)u_i(\bar s) + \sum_{\bar s_{-i}} p(s_i\bar s_{-i})u_i(s_i\bar s_{-i})
\end{align}

\begin{align}
\nonumber\mu_i(\Phi_i(\rho)) & = \sum_{\bar s\in S} \bra{\bar s} \Phi_i(\rho) \ket{\bar s} u_i(\bar s) \\
\nonumber& = \sum_{\bar s\in S} \bra{\bar s} \sum_{t_i \neq s_i} P_{t_i} \rho P_{t_i} + (s_i \leftrightarrow s_i') P_{s_i} \rho P_{s_i} (s_i \leftrightarrow s_i') \ket{\bar s} u_i(\bar s) \\
\nonumber& = \sum_{\bar s\in S} \bra{\bar s} \sum_{t_i \neq s_i} P_{t_i} \rho P_{t_i} \ket{\bar s} u_i(\bar s) + \sum_{\bar s\in S} \bra{\bar s} (s_i \leftrightarrow s_i') P_{s_i} \rho P_{s_i} (s_i \leftrightarrow s_i') \ket{\bar s} u_i(\bar s) \\
& = \sum_{t_i \neq s_i}\sum_{\bar s_{-i}} p(t_i \bar s_{-i}) u_i(t_i \bar s_{-i}) + \sum_{\bar s_{-i}} p(s_i \bar s_{-i}) u_i(s'_i\bar s_{-i}) \label{eqn:1}
\end{align}
% \eqref{eqn:1}
where in the last equality we used the fact that $P_{t_i}\ket{\bar s} = \ket{t_i \bar s_{-i}}$ if $\bar s_i = t_i$ and 0 otherwise; similar equality used for the second summand.

Since $\rho$ is a quantum correlated equilibrium, we have $\mu_i(\rho) \geq \mu_i(\Phi_i(\rho))$. Comparing the above two expressions for $\mu_i(\rho)$ and $\mu_i(\Phi_i(\rho))$ gives Eq. \eqref{eq: CE}, as desired.
\end{proof}

Many precious work try to find a quantum equilibrium with ``better" payoff than all classical ones, for example, to attempt to resolve the Prisoner's dilemma by showing a quantum equilibrium with payoff of both players better than the classical (unique) equilibrium. The theorem above implies that at least in our model, this is simply not possible. We actually think that this should be a property that reasonable quantization models should satisfy. 

\subsection{From classical to quantum: The classical mixture mapping and its conceptual implications}
The implication from classical to quantum turns out to be much more complicated. Let us consider the three types of mappings one by one. Recall that the first mapping $\rho(p) = \sum_s p(s) \ket{s}\bra{s}$ is the mixture of the classical states. The following theorem says that this always yields a quantum equilibrium from a classical equilibrium. That is, the utility $p_i(s)$ cannot be increased for a classical equilibrium even when player $i$ is allowed to have quantum operations.
\begin{Thm}[$p$ CE/NE $\Rightarrow \rho(p)$ QCE/QNE]\label{thm: classical mix}
	If $p$ is a (classical) correlated equilibrium, then $\rho(p) = \Sigma_{s\in S} p(s) \ket{s}\bra{s}$ is a quantum correlated equilibrium. In particular, if $p$ is a Nash equilibrium, then $\rho$ as defined is a quantum Nash equilibrium.
\end{Thm}
\begin{proof}
	Since the state $\rho(p)$ is essentially a classical one, whatever operation on $H_i$, followed by the measurement in the computational basis, only gives a new distribution over $S_i$ without affecting the distribution of $s_{-i}$. Since classically changing the given $s_i$ to any $s_i'$ does not increase the expected payoff, changing $s_i$ to a random $s_i'$ according to the new distribution does not give any advantage either. 
\end{proof}

The reason we still mention this technically trivial result is because it has a couple of conceptually important implications. First, together with Theorem \ref{thm: QE to CE}, it gives a one-one correspondence between classical Nash/correlated equilibria and a subset of quantum Nash/correlated equilibria. This can be used with Theorem \ref{thm: existence of NE} to answer the basic question of the existence of a quantum Nash equilibrium.

\begin{Cor}\label{cor: QNE-exist}
	Every game with a finite number of players and a finite set of strategies for each player has a quantum Nash equilibrium.
\end{Cor}

Second, one also notices that there is a one-one correspondence between the utility values in classical and quantum games. %That is, there is also a one-one correspondence between the joint utility values $(\av_{s\leftarrow p}[u_1(s)], \ldots, \av_{s\leftarrow p}[u_n(s)])$ for classical Nash/correlated equilibria $p$ and those $(u_1(\rho), \ldots, u_n(\rho)$ for quantum Nash/correlated equilibria. 
This immediately transfers all the \textbf{NP}-hardness results for finding an optimal Nash or correlated equilibrium \cite{GZ89} to the corresponding quantum ones.

Theorem \ref{thm: QE to CE} and \ref{thm: classical mix} also help to answer a basic question about the hardness of finding a quantum Nash equilibrium. One subtlety for quantum Nash equilibria is that it is a quantum state, so we need to first define what it means by ``finding" a quantum equilibrium: Is it sufficient to generate one, or to fully specify the state by giving all the matrix entries. It turns out that these two definitions are close to each other. 

\begin{Thm}\label{thm: QNE-PPAD}
Suppose that there is a polynomial-time quantum algorithm for finding a quantum Nash equilibrium $\rho$, with the guarantee that every execution of the algorithm gives the same $\rho$. Then there is a polynomial-time quantum algorithm to solve any problem in \ppad.	
\end{Thm}

%The forward direction is immediate by Theorem \ref{thm: classical mix}, because once we can find a Nash equilibrium $p$, then we find a quantum equilibrium $\rho(p)$. The backward direction is also not hard; 
Basically once having found a quantum Nash equilibrium $\rho$, one can use measurement in the computational basis to get a sample according to $p(\rho)$. Then taking an average of enough number of such samples gives a good enough (an inverse polynomial, to be precise) approximation, and then we can apply the hardness result of finding an approximate NE in \cite{CDT09}. Details are omitted.

\subsection{From classical to quantum: The quantum superposition mapping and its extremal properties}\label{sec: superposition}

The second way of inducing a quantum state from a classical distribution is by quantum superposition $\ket{\psi(p)} = \sum_s \sqrt{p(s)} \ket{s}$. This case is subtler than the classical mixture mapping: While an argument similar to that for Theorem \ref{thm: classical mix} shows that the quantum superposition mapping preserves Nash equilibrium property, it is not immediate to see whether it also does so for correlated equilibria. 

We consider to find the maximum incentive in two-player games, in which without loss of generality we can assume that the second player always getting payoff $1$. Indeed, any CE of any other bimatrix game $(A,B)$ is also a CE of $(A,J)$. We will formulate the maximum incentive finding problem over $n\times n$ bimatrix games $(A,J)$ by an optimization in Section \ref{sec:q-superposition opt} and give solution for the special case of $n=2$ in Section \ref{sec:q-superposition n=2}. Then for the general bimatrix games $(A,B)$, we will also consider $n\times n$ bimatrix game with Player 2's payoff being the all-one matrix, though our solutions are also CE for bimatrix game $(I_n,I_n)$, a natural extension of the Battle of the Sexes game.

\subsubsection{Maximum quantum incentive on $\ket{\psi(p)}$ as an optimization problem}\label{sec:q-superposition opt}
A CPTP operation $\Phi$ by Player 1 followed by the measurement in the computational basis $\{1, 2, \ldots, n\}$ gives a general POVM measurement $\{E_i: i\in [n]\}$. Suppose Player 1's payoff matrix is $A = [a_{ij}]$. Then Player $1$'s new payoff, \ie the payoff for playing $\Phi$, is 
\[
	\sum_{i,j\in [n]} a_{ij} \Big(\sum_{i_1\in [n]} \sqrt{p_{i_1 j}} \bra{i_1}\Big) E_i \Big(\sum_{i_2\in [n]} \sqrt{p_{i_2 j}} \ket{i_2} \Big)
\]
For simplicity let us use a short notation $\ket{\sqrt{p_j}}$ for $\sum_{i\in [n]}\sqrt{p_{i j}} \ket{i}$. Then the above payoff is $\sum_{i,j} a_{ij} \bra{\sqrt{p_j}} E_i \ket{\sqrt{p_j}}$. Thus the maximum quantum additive incentive on $\ket{\psi(p)}$ for a CE $p$ can be written as the following optimization problem. 
\begin{align*}
\textbf{Primal: } \qquad  \max & \quad \sum_{i,j\in [n]} a_{ij} (\bra{\sqrt{p_j}} E_i \ket{\sqrt{p_j}}  - p_{ij}) \\
	\mbox{s.t.} & \quad 0\leq a_{ij} \leq 1, \quad \forall i,j\in [n] & (\text{The game is [0,1]-normalized.})\\
	& \quad \sum_{ij}p_{ij} = 1, \quad p_{ij} \geq 0, \quad \forall i,j\in [n] & (p \text{ is a distribution.})\\
	& \quad \sum_j a_{ij}p_{ij} \geq \sum_j a_{i'j}p_{ij}, \quad \forall i,i',j\in [n] &  (p \text{ is a correlated equilibrium.})\\
	& \quad \sum_i E_i = I_n, \quad E_i \succeq 0, \quad \forall i\in [n] & (\{E_i\} \text{ is a POVM measurement.})
\end{align*}
And the maximum quantum multiplicative incentive is the same except the objective function now becomes $(\sum_{i,j\in [n]} a_{ij} \bra{\sqrt{p_j}} E_i \ket{\sqrt{p_j}} ) / (\sum_{i,j\in [n]} a_{ij}p_{ij}))$.

Note that the objective function is highly non-concave\footnote{Sometimes people say convex programming for convex minimization, or equivalently as in our case, concave maximization.}, which makes the problem generally hard to compute or analyze. (The non-concavity can be witnessed by the optimal solution of the case of $n = 2$ shortly.) One way for handling this is to fix some of the variables and consider the dual of the remaining problem. If we fix $A = [a_{ij}]$ and $P = [p_{ij}]$, then it is a semi-definite program with variable $E_i$'s. The dual of it is the following.

\begin{align*}
	\textbf{Dual}(A,P): \qquad \qquad \min & \quad Tr(Y) - \sum_{i,j\in [n]}a_{ij}p_{ij} \\
	\mbox{s.t.} & \quad Y \succeq \sum_{j\in [n]} a_{ij} \ket{\sqrt{p_j}} \bra{\sqrt{p_j}}, \quad \forall i\in [n]
\end{align*}
One can also write down the dual for the multiplicative incentive optimization primal by simply changing the subtraction to division in the objective function; note that it is still linear in $E_i$'s for fixed $A$ and $P$. Sometimes working with dual helps to establish the optimality of the objective function value on a primal feasible solution that we find. 

\subsubsection{Complete solution of $2\times 2$ games}\label{sec:q-superposition n=2}
We first study $2\times 2$ games, which turns out to be nontrivial already, and the experiences we obtain here will be useful later for general games. First it is not hard to see that in an optimal solution, all $a_{ij}$'s are either 0 or 1. Then one can see that only $A = I_2$ or $A = \begin{bmatrix} 0 & 1 \\ 1 & 0\end{bmatrix}$ may admit positive incentive. Since permuting columns or rows will not change the optimal value, let us assume that $A = I$ in the following. We do not know whether $A=I$ is also a maximizer for the general problem; it is an interesting open question.

The equilibrium property implies that $p_{11}\geq p_{12}$ and $p_{22} \geq p_{21}$. By the requirement $E_1 \succeq 0$ and $E_2 = I-E_1 \succeq 0$, we can assume that the optimal value $E_1 = \begin{bmatrix} a & c \\ c^* & 1-b \end{bmatrix}$, where $a, b \in [0,1]$, and $|c|^2 \leq \min\{a(1-b), b(1-a)\}$. Then the primal value is 
\begin{align}
	& \ \bra{\sqrt{p_1}} E_1 \ket{\sqrt{p_1}} + \bra{\sqrt{p_2}} E_2 \ket{\sqrt{p_2}} - p_{11} - p_{22} \\
	= & \ 2\cdot Re(c) \cdot (\sqrt{p_{11}p_{21}} - \sqrt{p_{12}p_{22}}) - (p_{11}-p_{12})(1-a) - (p_{22}-p_{21})(1-b) \\
	= & \ 2\cdot |Re(c)| \cdot \big|\sqrt{p_{11}p_{21}} - \sqrt{p_{12}p_{22}}\big| - (p_{11}-p_{12})(1-a) - (p_{22}-p_{21})(1-b).
\end{align}
where the last equality is because the optimal value is nonnegative and thus 
\begin{equation}
	Re(c) \cdot (\sqrt{p_{11}p_{21}} - \sqrt{p_{12}p_{22}}) \geq (p_{11}-p_{12})(1-a) + (p_{22}-p_{21})(1-b) \geq 0.
\end{equation} Further, in a maximizer, $|Re(c)|$ should be as large as possible, so it holds that $c\in \mbR$ and either $c^2 = a(1-b)$ or $c^2 = b(1-a)$. We claim that actually both hold and thus $a = b$. Actually, if $|c|^2 = a(1-b)<b(1-a)$, then $a<b$, and it can be observed that the objective function increases with $a$. So one can increase $a$ up to $b$; the other case of $a>b$ can be argued in the same way. 

Now the primal value becomes 
\begin{equation}
	2\sqrt{a(1-a)} \cdot \big|\sqrt{p_{11}p_{21}} - \sqrt{p_{12}p_{22}}\big| - (p_{11}-p_{12} + p_{22}-p_{21})(1-a).
\end{equation}
By simultaneously switching the two rows and columns, one can assume that $p_{11}p_{21} \leq p_{12}p_{22}$. (We need to switch rows and columns simultaneously because we have already assumed the matrix to be $I$.) Then the optimal value is 
\begin{align}
	OPT & = 2\sqrt{a(1-a)} \cdot (\sqrt{p_{12}p_{22}} - \sqrt{p_{11}p_{21}}) - (p_{11}-p_{12} + p_{22}-p_{21})(1-a) \\
	& \leq 2\sqrt{a(1-a)} \cdot \left(\sqrt{\frac{p_{11}+p_{12}}{2}p_{22}} - \sqrt{\frac{p_{11}+p_{12}}{2}p_{21}}\right) - (p_{22}-p_{21})(1-a)
\end{align}
That is, we shift mass from $p_{11}$ to $p_{12}$ and the objective function always increases. This can be done as long as the equilibrium properties is maintained, namely $p_{11} \geq p_{12}$. Since $P$ is maximizer, we know that $p_{11} = p_{12}$. Thus 
\begin{align}
	OPT & = (2\sqrt{a(1-a)} \sqrt{p_{11}} - (1-a)(\sqrt{p_{22}} + \sqrt{p_{21}})) (\sqrt{p_{22}} - \sqrt{p_{21}}) \\
	& \leq \big(2\sqrt{a(1-a)} \sqrt{p_{11}} - (1-a)(\sqrt{p_{22} + p_{21}})\big) \sqrt{p_{22}}
\end{align}
Thus if we shift mass from $p_{21}$ to $p_{22}$, then the objective function value increases. So the maximizer $p$ has $p_{21} = 0$, and we have
\begin{align}
	OPT & = 2\sqrt{a(1-a)} \sqrt{p_{11}(1-2p_{11})} - (1-a)(1-2p_{11})
\end{align}
Now by looking at the partial derivative (and setting it to be zero), it is not hard to finally find that $p_{11}^* = \sqrt{2}/4$ and $a^* = \sqrt{2}/2$ give the maximum value $(\sqrt{2}-1)/2$, which is the maximum quantum additive incentive. The corresponding optimal solutions for the primal and the dual are as follows. 
\begin{align}
	\text{Additive OPT}: \qquad & (\sqrt{2}-1)/2 = 0.2071...\\
	\text{Primal solution}: \qquad & P = \begin{bmatrix} p_{11}^* & p_{11}^* \\ 0 & 1-2p_{11}^* \end{bmatrix}, \\
	& E_1 = \begin{bmatrix} 2p_{11}^* & -\sqrt{2p_{11}^*(1-2p_{11}^*)} \\ -\sqrt{2p_{11}^*(1-2p_{11}^*)} & 1-2p_{11}^* \end{bmatrix}, \quad E_2 = I - E_1, \\
	\text{Dual solution}: \qquad & Y = \begin{bmatrix} 1/2 & \sqrt{p_{11}^*(1/2-p_{11}^*)} \\ \sqrt{p_{11}^*(1/2-p_{11}^*)} & p_{11}^* \end{bmatrix}.
\end{align}

\vspace{1em}
The solution also confirms that the objective function of the Primal for additive incentive is not concave. Indeed, by symmetry, another optimal solution for Primal is 
\begin{equation}
	P' = \begin{bmatrix} 1-2p_{11}^* & 0 \\ p_{11}^* & p_{11}^* \end{bmatrix}, \quad E_1' = \begin{bmatrix} 2p_{11}^* & \sqrt{2p_{11}^*(1-2p_{11}^*)} \\ \sqrt{2p_{11}^*(1-2p_{11}^*)} & 1-2p_{11}^* \end{bmatrix}, \quad E_2' = I - E_1'.
\end{equation}
But the average of the two solutions gives a negative objective value. 

One may wonder whether the objective function is concave ``with respect to" $p$, that is, if we are allowed to take optimal $E$ for each $p$. Unfortunately it is still not concave: Actually for $(P+P')/2$ there is not any positive incentive, as can be witnessed by the dual matrix 
\begin{equation}
	Y = \begin{bmatrix} (1-p_{11}^*)/2 & \sqrt{p_{11}^*(1-p_{11}^*)}/2 \\ \sqrt{p_{11}^*(1-p_{11}^*)}/2 & (1-p_{11}^*)/2 \end{bmatrix}.
\end{equation}
It can be easily verified that $Y$ is a feasible solution for the dual, and it gives the value $Tr(Y) - Tr(P) = 0$, which is an upper bound of the optimal value for this $(P+P')/2$.

Using a similar method, one can also find that the maximum quantum multiplicative incentive is 4/3. The optimal solutions of the primal and dual are as follows. 
\begin{align}
	\text{Multiplicative OPT}: \qquad & 4/3, \\
	\text{Primal solution}: \qquad & P = \begin{bmatrix} 2/5 & 2/5 \\ 0 & 1/5 \end{bmatrix}, \\
	& E_1 = \begin{bmatrix} 2/3 & -\sqrt{2}/3 \\ -\sqrt{2}/3 & 1/3 \end{bmatrix}, \quad E_2 = I - E_1, \\
	\text{Dual solution}: \qquad & Y = \begin{bmatrix} 8/15 & 2\sqrt{2}/15 \\ 2\sqrt{2}/15 & 4/15 \end{bmatrix}.
\end{align}
We have then completely solved the case of $n=2$.

\subsubsection{Lower bounds for general games}
Next we will study game of the general size $n$ and prove first part of Theorem  \ref{thm:main1}. Note that the \emph{ad hoc} analysis used in previous part cannot be generalized in any straightforward way to the general case. However, some insights obtained there are useful in the later construction. 
%\begin{Thm}[$p$ NE $\Rightarrow \ket{\psi(p)}$ QNE, $p$ CE $\nRightarrow \ket{\psi(p)}$ QCE]\label{thm: superposition}
%If $p$ is a classical Nash equilibrium, then $\ket{\psi(p)}$ is a quantum Nash equilibrium. On the other hand, there exists a classical correlated equilibrium $p$ with $\ket{\psi(p)} = \sum_{s} \sqrt{p(s)}\ket{s}$ not being a quantum correlated equilibrium. In a normalized game, the maximum quantum additive incentive is at least $1-O(\log^2\log(n)/\log(n))$ and the maximum quantum multiplicative incentive is at least $n^{\log_2(4/3)} = n^{0.4150\ldots}$. 
%\end{Thm}
%\begin{proof}
%The Nash equilibrium case can be easily argued by an argument similar to that for Theorem \ref{thm: classical mix}. Indeed, any operation on $H_i$ followed by the measurement in the computational basis gives a new distribution $p_i'$ over $S_i$ without affecting $p_{-i}$, the distribution of $s_{-i}$. Then by the assumption that $p$ is a Nash equilibrium, changing $p_i$ to any $p_i'$ does not increase the expected payoff. 
%The more interesting case is the correlated equilibria.

We will exhibit a family of games and correlated equilibria $p$ such that the quantum incentive in $\ket{\psi(p)}$ increases with the size of the game. Before giving the construction, let us briefly discuss the intuition. Suppose we already have a small game matrix $A$ and a correlated equilibrium $p$ with positive quantum additive incentive on $\ket{\psi(p)}$. How to construct a larger game with a larger quantum additive incentive? Note that we are to find a distribution $p'$ satisfying two requirements: First, it is a CE of the larger game, and second, $\ket{\psi(p')}$ has a larger quantum additive incentive. It turns out that tensor product can satisfy both properties if the parameters are good. 

%A naive approach is to tile the game matrix and the CE matrix, getting the game $\begin{bmatrix} A & A \\ A & A \end{bmatrix}$ and the distribution $\begin{bmatrix} p/4 & p/4 \\ p/4 & p/4 \end{bmatrix}$. While the new distribution does satisfy the equilibrium property, the additive incentive is the same as before and the multiplicative incentive 

\begin{Lem}
	For two bimatrix games $(A_1,B_1)$ and $(A_2,B_2)$ with two correlated equilibria $p_1$ and $p_2$ (of the two games respectively), suppose Player 1's expected payoff on $\ket{\psi(p_i)}$ is $u_i$, and her maximum quantum additive and multiplicative incentives on $\ket{\psi(p_i)}$ are $a_i$ and $m_i$. Then the distribution $p_1\otimes p_2$ is a correlated equilibrium of the larger game $(A_1\otimes A_2, B_1\otimes B_2)$, and the maximum quantum additive and multiplicative incentives on $\ket{\psi(p_1\otimes p_2)}$ are at least $(u_1+a_1)(u_2+a_2)-u_1u_2$ and $m_1m_2$, respectively.
\end{Lem}
\begin{proof}
Let us first show that $p_1\otimes p_2$ is a correlated equilibrium of $(A_1\otimes A_2, B_1\otimes B_2)$. Given any strategy $x\circ y$ of Player 1, where $x$ and $y$ are two strategies of Player 1 in games $(A_1, B_1)$ and $(A_2, B_2)$, respectively, the conditional distribution of Player 2's strategy is $p_1|_x \times p_2|_y$, where $p_1|_x$ is the distribution of Player 2's strategy in game $(A_1, B_1)$ conditioned on Player 1 getting $x$ (in a sample from $p_1$), and similarly for $p_2|_y$. Note that $p_1|_x \times p_2|_y$ is a product distribution, therefore, Player 1 changing the strategy to any other $x' \circ y'$ does not increase her expected payoff, since the expectation decomposes as the product of two expectations in the two small games, both cannot be increased by changing strategies by the definition of correlation equilibrium.

Now we calculate the payoffs. The average payoff of Player 1 in $(A_1\otimes A_2, B_1\otimes B_2)$ for strategy $p_1\otimes p_2$ is 
\begin{align}
	\langle p_1\otimes p_2, A_1\otimes A_2 \rangle = \langle p_1, A_1 \rangle \cdot \langle p_2, A_2 \rangle = u_1 u_2
\end{align}
If the maximum quantum multiplicative incentive on $\ket{\psi(p_i)}$ are achieved by Player 1 applying $\Phi_{i}$, then the maximum quantum multiplicative incentive on $\ket{\psi(p_1\otimes p_2)}$ is at least $m_1m_2$, since Player 1 can at least apply the local operation $\Phi_1\otimes \Phi_2$. The additive incentive on $\ket{\psi(p_1\otimes p_2)}$ follows similarly.
\end{proof}
 
We want to use the lemma $d = \lfloor \log_c(n)\rfloor$ times, recursively, to construct a game of size $n$ from small building-block games of size $c$. First consider $c = 2$ and $n$ being a power of 2; same asymptotic bound holds for general $n$ (by looking at the largest submatrix of size $2^d$). Some experiences from the last section, such as $A = I$ and $p_{21} = 0$, help to design the $2\times 2$ game. But note that simply taking the optimal solution of the $2\times 2$ games will not work since eventually the quantum additive incentive will be $(u_1(\Phi_1(\rho)))^d - (u_1(\rho))^d = (\sqrt{2}/4+1/2)^d - (1-\sqrt{2}/4)^d = o(1)$. To have the additive incentive $(u_1(\Phi_1(\rho)))^d - (u_1(\rho))^d$ large, it needs $u_1(\Phi_1(\rho))$ to be very close to 1. It turns out that for a small-size game $(I_2,J_2)$, if $u_1(\Phi_1(\rho))$ is close to 1, so is $u_1(\rho)$. Thus the incentive in the size-$c$ game is actually very small, far from being a good solution of the small game. 

With this in mind, we construct the game in the following way. 
%$2\times 2$ bimatrix game and a CE $p$, such that Player 1 has no classical incentive to deviate in $p$, but a positive quantum incentive to deviate in the induced $\ket{\psi(p)}$. We will then construct a larger game by tensor product to achieve the larger incentive. 
Again define utility functions of Player 1 and 2 
\begin{equation}
		A_1 = \begin{bmatrix} 1 & 0 \\ 0 & 1 \end{bmatrix}, \quad B_1 = \begin{bmatrix} 1 & 1 \\ 1 & 1 \end{bmatrix},
\end{equation}
and a probability distribution 
\begin{equation}\label{eq: best addi P for n=2}
 P = \begin{bmatrix} \sin^2(\epsilon) & \cos^2(\epsilon)\sin^2(\epsilon) \\ 0 & \cos^4(\epsilon)	\end{bmatrix},
\end{equation} 
where $\epsilon$ is a small number to be decided later. It is not hard to verify that $p$ is a CE with Player 1's average utility being 
\begin{equation}
	\mu_{1,old} = tr(P) = \sin^2(\epsilon) + \cos^4(\epsilon).
\end{equation}

The induced quantum superposition state
\begin{equation}
	\ket{\psi(p)} = \sin(\epsilon)\ket{00} + \cos(\epsilon)\sin(\epsilon)\ket{01} + \cos^2(\epsilon)\ket{11},
\end{equation}
is not a QCE, because Player 1 can apply the unitary operator
\begin{equation}
	U_1 = \begin{bmatrix} \cos(\epsilon) & -\sin(\epsilon) \\ \sin(\epsilon) & \cos(\epsilon) \end{bmatrix},
\end{equation}
which has a general effect of
\begin{equation}
	\begin{matrix}
		\ \ \ \cos(a)\cos(b)\ket{00} + \sin(a)\cos(c)\ket{01} \\
		+ \cos(a)\sin(b)\ket{10} + \sin(a)\sin(c)\ket{11}
	\end{matrix}
	\quad \rightarrow \ 
	\begin{matrix}
	\ \ \ \cos(a)\cos(b+\epsilon)\ket{00}  + \sin(a)\cos(c+\epsilon)\ket{01} \\
	+ \cos(a)\sin(b+\epsilon)\ket{10} + \sin(a)\sin(c+\epsilon)\ket{11}
	\end{matrix}.
\end{equation}
So applying $U_1$ on $\ket{\psi(p)}$ gives  
\begin{equation}
	U_1\ket{\psi(p)} = \sin(\epsilon)\cos(\epsilon)\ket{00} + \sin^2(\epsilon)\ket{10} + \cos(\epsilon)\ket{11},
\end{equation}
which has a utility of 
\begin{equation}
	\mu_{1,new} = \sin^2(\epsilon)\cos^2(\epsilon) + \cos^2(\epsilon).
\end{equation}
%The quantum multiplicative incentive is thus $q_1 = $.

Now we apply the above lemma to define a large game by $A_d = A^{\otimes d}, \quad B_d = B^{\otimes d}$ and a correlated equilibrium by $P_d = P^{\otimes d}$. Recall that $d = \lfloor \log_2 n\rfloor$. Let $\epsilon$ satisfy $d = 4\epsilon^{-2}\ln(1/\epsilon)$, which gives $\epsilon = \Theta(\sqrt{\log d/d})$. Using the above tensor product construction, we get a quantum additive incentive of
\begin{align}
	\mu_{1,new}^d - \mu_{1,old}^d & = (\sin^2(\epsilon)\cos^2(\epsilon) + \cos^2(\epsilon))^d - (\sin^2(\epsilon) + \cos^4(\epsilon))^d \\
	& = (1-\sin^4(\epsilon))^d - (1- \sin^2(2\epsilon)/4)^d \\
	& \geq (1-d\epsilon^4) - (e^{-d\sin^2(2\epsilon)/4}) \\
	& = 1 - (4\epsilon^2\ln\frac{1}{\epsilon} + \epsilon^{\epsilon^{-2}\sin^2(2\epsilon)}) \\
	& \geq 1 - (4\epsilon^2\ln\frac{1}{\epsilon} + \epsilon^{4-16\epsilon^2/3}) \\
	& = 1-O\Big(\frac{\log^2d}{d}\Big) = 1-O\Big(\frac{\log^2\log n}{\log n}\Big)
\end{align}
where the first inequality used the bounds $\sin(x) < x$ and $1-dx < (1-x)^d < e^{-dx} $, for any $x>0$, and the second inequality used the bound $\sin(x) \geq x - x^3/6$ for $x>0$. 

\medskip
For the multiplicative incentive, we can simply take the $2\times 2$ game with the maximum multiplicative quantum incentive, 4/3, in the last section. The resulting multiplicative quantum incentive is then $(4/3)^{\log_2 n} = n^{\log_2(4/3)} = n^{0.4150\ldots}$. This falls short of the promise in Theorem \ref{thm:main1}. We now give another construction for general dimension $c$, which yields a better multiplicative incentive. Consider the following $c\times c$ bimatrix game. The utility function is still $A_1 = I_c$, $B_1 = J_c$, and define a distribution $p$ by 
\begin{equation}\label{eq: best multi P for n=2}
	p_{ij} = \begin{cases} 0 & i-j = 1 \text{ mod } c \\ \frac{1}{c^2-c} & \text{otherwise} \end{cases}
\end{equation}
where $i,j$ range over $\{0, 1, \ldots, c-1\}$. It is routine to check that it is indeed a correlated equilibrium, and Player 1's current utility is $c\cdot 1/(c^2-c) = 1/(c-1)$. Let the POVM $\{E_0, \ldots, E_{c-1}\}$ be 
\begin{equation}
	E_i = \ket{\psi_i} \bra{\psi_i}, \text{ where the } i' \text{-th entry of vector } \ket{\psi_i} \text{ is } \psi_{i,i'} = \begin{cases} \frac{2-c}{c} & i' - i = 1 \text{ mod } c \\ \frac{2}{c} & \text{otherwise} \end{cases}.
\end{equation}
It is a valid POVM measurement: 
\begin{equation}
	\sum_i E_i(j,j) = \left(\frac{2-c}{c}\right)^2 + (c-1)\left(\frac{2}{c}\right)^2 = 1, \quad \forall j,
\end{equation}
and 
\begin{equation}
	\sum_i E_i(j,j') = 2\cdot \frac{2-c}{c}\cdot \frac{2}{c} + (c-2)\left(\frac{2}{c}\right)^2 = 0, \quad \forall j\neq j'.
\end{equation}
Now the new probability is 
\begin{equation}
	p'_{ij} = \sum_{i1, i2} \sqrt{p_{i_1,j}p_{i_2,j}}E_i(i_1,i_2).
\end{equation}
and the new utility is 
\begin{align}
	& \ \sum_i \sum_{i1, i2} \sqrt{p_{i_1,i}p_{i_2,i}}E_i(i_1,i_2) \\
	= & \ \sum_i \sum_{i_1 \neq i+1, i_2 \neq i+1} \frac{1}{c^2-c} \left(\frac{2}{c}\right)^2\\
	= & \ c\cdot (c-1)^2 \cdot \frac{4}{c^3(c-1)} \\
	= & \ \frac{4(c-1)}{c^2}
\end{align}
So the multiplicative incentive is $4(c-1)^2/c^2$. By the same tensor product construction, we get an $n\times n$ game with multiplicative incentive \begin{equation}
	\Big(\frac{4(c-1)}{c^2}\Big)^{\log_c(n)} = n^{\frac{2+2\log_2(1-1/c)}{\log_2 c}}.
\end{equation}
Optimizing this over integers $c$, we get a quantum multiplicative incentive $n^{\log_2(3)-1} = n^{0.585...}$ at $c = 4$. 

\medskip
A final remark for this section is that both lower bounds, for the maximum additive and multiplicative incentives, can be achieved even by symmetric games. Indeed, it is not hard to verify that the probability distributions $P^{\otimes d}$ with $P$ given in Eq. \eqref{eq: best addi P for n=2} and Eq. \eqref{eq: best multi P for n=2} are still correlated equilibria for the game $(I,I)$, a natural extension of \emph{Battle of the Sexes} game in Section \ref{sec: intro}. 

\subsection{From classical to quantum: General mappings and their extremal properties}
Finally, for the general mapping, \ie an arbitrary quantum state $\rho$ with $p(s) = \rho_{ss}$ satisfied, the equilibrium property can be heavily destroyed, even if $p$ is uncorrelated. We can pin down the exact maximum quantum additive and multiplicative incentives.

\begin{Thm}[$p$ NE $\nRightarrow \rho$ QCE]\label{thm: examples}
	There exist $\rho$ and $p$ satisfying that $p(s)= \rho_{ss}$, $p$ is a Nash equilibrium, but $\rho$ is not even a quantum \emph{correlated} equilibrium. The maximum quantum additive incentive in a normalized $(m\times n)$-bimatrix game is $1-1/\min\{m,n\}$, and the maximum multiplicative incentive is $\min\{m,n\}$ even for correlated equilibria $p$.
\end{Thm}

The theorem is a corollary of the next more general one.
\begin{Thm}
	Suppose $p$ is a \emph{correlated} equilibrium for a normalized $n$-player game and $\rho$ satisfies $\rho_{ss} = p(s)$, $\forall s\in S$. Then the maximum quantum additive incentive is at most $1-\epsilon_i$ and the maximum quantum multiplicative incentive is at most $1/\epsilon_i$, where $\epsilon_i = \max\{|S_i|^{-1}, |S_{-i}|^{-1}\}$. Both bounds are achievable even by some \emph{Nash} equilibrium $p$.
\end{Thm}
\begin{proof}
Suppose Player $i$ applies operation $\Psi_i$ on $\rho$, resulting in a distribution $\lambda$ on $S$ when the players measure the state in the computational basis. Sine local operation cannot change other parties' density operator, the marginal distribution of $\lambda$ on $S_{-i}$ is still $p_{-i}$. The new payoff for Player $i$ is $\|u\circ \lambda\|_1$, where $\|\cdot \|_1$ is the sum of entries in absolute value. We are going to prove that the original payoff for Player $i$ is at least $\epsilon_i$ fraction of the new payoff; that is,
\begin{equation}
	\|u\circ p\|_1 \geq \epsilon_i \|u\circ \lambda\|_1.
\end{equation}
This would imply the claimed bound for multiplicative incentive, and the additive incentive follows: $(1-\epsilon_i)\|u\circ \lambda\|_1 \leq 1-\epsilon_i$ since $u$ is normalized. 

Now we prove the above inequality. First consider the case of $\epsilon_i = |S_{-i}|^{-1}$. For each $s_{-i}\in S_{-i}$, define a probability distribution $p_i^{s_{-i}}$ over $S_i$ by  $p_i^{s_{-i}}(s_i) = \lambda(s_is_{-i})/p_{-i}(s_{-i})$. Then 
\begin{equation}
	\|u \circ (p_i^{s_{-i}}\times p_{-i})\|_1 = \sum_{s_i, s_{-i}'} \frac{\lambda(s)}{p_{-i}(s_{-i})}u(s_i, s_{-i}') p_{-i}(s_{-i}') \geq \sum_{s_i} \lambda(s_is_{-i})u(s_i, s_{-i})  
\end{equation}
where in the last step we dropped the summands for all $s_{-i}' \neq s_{-i}$. Now define 
\begin{equation}
	\bar p_i = \frac{1}{|S_{-i}|} \sum_{s_{-i}\in S_{-i}} p_i^{s_{-i}},
\end{equation}
the average the these distributions $p_i^{s_{-i}}$. Since $p$ is a correlated equilibrium, Player $i$ cannot increase her expected payoff by switching to $\bar p_i$. So 
\begin{align}
	\|u \circ p\|_1 & \geq \|u \circ (\bar p_i p_{-i})\|_1 = \frac{1}{|S_{-i}|} \sum_{s_{-i}} \|u \circ (p_i^{s_{-i}} p_{-i})\|_1 \\
	& \geq \frac{1}{|S_{-i}|} \sum_{s_i,s_{-i}} \lambda(s_is_{-i})u(s_i, s_{-i}) = \frac{1}{|S_{-i}|} \|u\circ \lambda\|_1. 
	%\|u, p_i^{s_{-i}}p_{-i}\|_1 = \sum_{s_i, s_{-i}'} \frac{\lambda(s)}{p_{-i}(s_{-i})}u(s_i, s_{-i}') p_{-i}(s_{-i}') \geq \sum_{s_i} \lambda(s)u(s_i, s_{-i})  
\end{align}
where the first equality is by noting that all matrices here are nonnegative.

For the case of $\epsilon_i = |S_{i}|^{-1}$, take the uniform distribution $q_i$ over $S_i$, then by the similar argument as above, we have 
\begin{align}
	\|u \circ p\|_1 \geq \|u \circ (q_i\times p_{-i})\|_1 = \frac{1}{|S_{i}|} \sum_{s_{i},s_{-i}} u(s_{i}s_{-i}) p_{-i}(s_{-i}).
\end{align}
Now note that $p_{-i}$ is the marginal distribution of $\lambda$ on $S_{-i}$, thus $p_{-i}(s_{-i}) \geq \lambda(s_is_{-i})$, and
\begin{align}
	\|u \circ p\|_1 \geq \frac{1}{|S_{i}|} \sum_{s_{i},s_{-i}} u(s_{i}s_{-i}) \lambda(s_{i}s_{-i}) = \frac{1}{|S_{i}|} \|u\circ \lambda\|_1. 
\end{align}
as desired.

We next show that the bounds in the above theorem is achievable even by a Nash equilibrium $p$. Assume that $|S_i| = |S_{-i}| = n$, then there is a one-one correspondence $\pi: S_{-i} \rightarrow S_i$. Consider the following $n$-player game: 
\begin{equation}
	u_i(s) = \begin{cases} 1 & \text{ if } s_i = \pi(s_{-i}) \\ 0 & \text{ otherwise }\end{cases}, \qquad u_j(s) = 1, \quad \forall j\neq i.
\end{equation}

Consider the state
\begin{equation}
	\ket{\psi} = (F_{S_i}\otimes I_{S_{-i}})\ket{\psi'}, \quad \text{ with } \quad \ket{\psi'} = \frac{1}{\sqrt{n}} \sum_{s_{-i} \in S_{-i}} \ket{\pi(s_{-i}) s_{-i}} 
\end{equation}
where $F_{S_i}$ is the Fourier transform operator on the register corresponding to $S_i$ (and $I_{S_{-i}}$ is the identity on the rest). %It is easy to see that $\ket{\psi}$ is a superposition state induced by the uniform distribution on $S$. 
If we measure $\ket{\psi}$, then we get a uniform distribution over the $n^2$ joint strategies. This is a Nash equilibrium, since if all other $[n]-\{i\}$ players choose a random strategy in $S_{-i}$, then Player $i$ is indifferent in all her $n$ strategies in $S_i$. 

However, $\ket{\psi}$ is not a quantum (even correlated) Nash equilibrium, because Player $i$ can apply the inverse Fourier transform on $\ket{\psi}$ to get $\ket{\psi'}$, which gives Player $i$ payoff 1 if the players measure the state. The gained payoff by this local operation is $1-1/n = 1-\epsilon_i$. 
\end{proof}

\section{Separation in classical and quantum correlation complexity of correlated equilibria}
This section studies the correlation from its generation. First observe that all correlations are correlated equilibria for some game. Actually, for any given probability distribution $p$ on $S$, for any $s_i$, let 
\begin{equation}
	s_{-i}^* =  \text{ the lexicographically first maximizer for } \max_{s_{-i}} p(s_is_{-i}).
\end{equation}
Define the utility function to be 
\begin{equation}
	u_i(s) = \begin{cases} 1 & \text{ if } s_{-i} = s_{-i}^* \\ 0 & \text{ otherwise }\end{cases}.
\end{equation}
Then it is easy to verify that $p$ is a correlated equilibrium. Thus the problem of generating correlated equilibria is as general as that of generating an arbitrary correlation.

Consider the following scenario for correlation generation. Two parties, $\alice$ and $\bob$, share some ``seed" correlation initially, and then perform local operations on their own systems. Different resources can serve as the seed correlation; in particular, it can be shared (classical) randomness and entangled (quantum) states. 

In the setting of games, two scenarios can be considered, depending on whether the local operations are carried out by trusted parties or untrusted players. We will discuss these models in the next two subsections.

\subsection{Correlated equilibrium generation: trusted local operation model}
\begin{figure}%
\begin{center}
\includegraphics[width=4in]{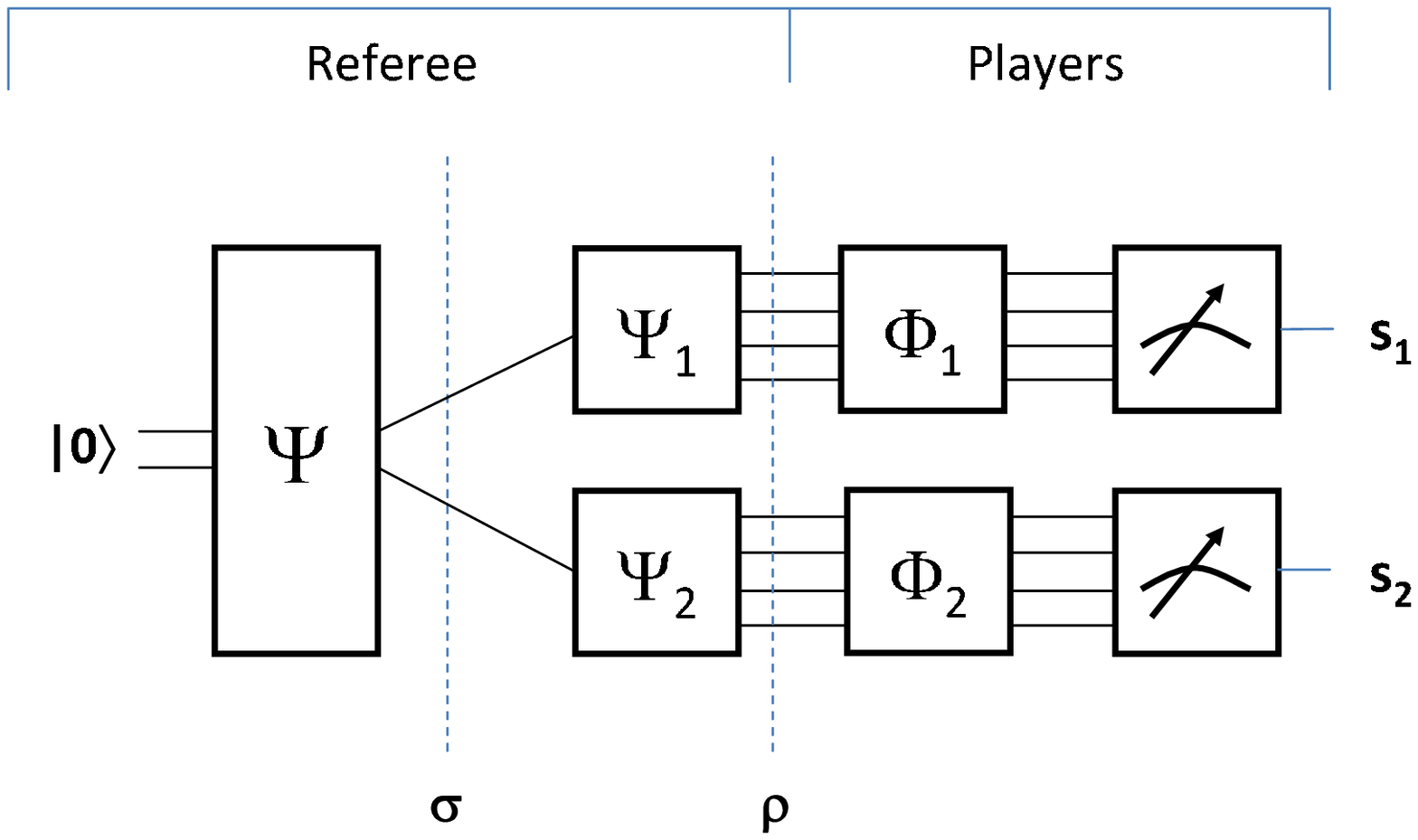}%
\end{center}
\caption{Correlated equilibrium generation with trusted local operations}%
\label{fig: trusted}%
\end{figure}

To illustrate the trusted local operation model, consider the a generalized \emph{Battle of the Sexes} game, where \alice and \bob are not in the same city but want to generate some correlation $p=(X,Y)$. There is a publicly trusted company C, which can help to generate $p$. Company C has a central server which generates a seed and send to its local servers A and B, distributed close to \alice and \bob, respectively. The local servers A and B apply the local operations to generate a state which is then sent to \alice and \bob. Here the local operations are carried out by the trusted servers A and B. And the complexity that we care is the size of the seed, which is also the communication between the central server to the two distributed servers A and B. 

More precisely, in the classical case, the two parties \alice and \bob initially have random variables $S_A$ and $S_B$, respectively, which may be correlated in an arbitrary way. They can also use private randomness $R_A$ and $R_B$, respectively. The two parties then apply local operations on their own systems. The joint output is then a pair of (correlated) random variables $(X,Y)$ where $X = f_A(S_A,R_A)$ and $Y = f_B(S_B,R_B)$ for some functions $f_A$ and $f_B$. In the quantum setting, the two parties initially share a state $\rho$, and they then apply local operations and output a pair of classical random variables $(X,Y)$. 
\begin{Def}
	The randomized correlation complexity of a distribution $p$ is the minimum size of shared random variables $(X',Y')$ given which \alice and \bob can apply local operations (but no communications) and output $X$ and $Y$, respectively, such that $(X,Y)$ is distributed according to $p$. The quantum correlation complexity is defined in the same way with the initially shared $(X',Y')$ being a quantum entangled state. We use $\rcorr(p)$ and $\qcorr(p)$ to denote the randomized and quantum correlation complexity of $p$. 
\end{Def}

We can also define the private-coin randomized (and quantum, respectively) communication complexity of distribution $p$, which is the minimum number of bits (and qubits, respectively) exchanged such that at the end of the protocol, \alice outputs $X$ and \bob outputs $Y$ with $(X,Y)$ distributed according to $p$. Note that no seed correlation is allowed in this case; that is why we call it private-coin. We use $\rcomm(p)$ and $\qcomm(p)$ to denote the private-coin randomized and quantum communication complexity of $p$. 

Some remarks are in order. First, recall that the size of the seed correlation $(X',Y')$ is half of the number of bits of $(X',Y')$, consistent with the convention that the size of public-coin string is the number of bits of $R$ which \alice and \bob share. Second, since (even one-way) communication can easily simulate the shared randomness/entanglement (by one party generating the shared resource and sending part of it to the other party), we have $\rcomm(p) \leq \rcorr(p)$ and $\qcomm(p) \leq \qcorr(p)$. It turns out that actually equality holds in both cases. However we still define the correlation complexity because it is a natural model and it is easier to bound (for example, in the later Theorem \ref{thm:qcorr}). Third, as we mentioned, \alice and \bob can always share the target correlation as the seed, so $\qcorr(p) \leq \rcorr(p) \leq \size(p)$. Finally, using a round-by-round argument, one can prove that $\qcomm(p) \geq I(p)/2$ where $I(p)$ is the mutual information $I(X,Y)$ for $(X,Y)\leftarrow p$. Putting all these together, we have
\begin{equation}
	\frac{I(p)}{2} \leq \qcomm(p) = \qcorr(p) \leq \rcomm(p) = \rcorr(p) \leq \size(p).
\end{equation}

\paragraph{Remark} The $\qcomm(p) = \qcorr(p)$ was pointed out firstly by Nayak (private communication), who observed that the argument in Kremer's thesis \cite{Kre95} (which was in turn attributed to Yao) implies that the Schmidt rank of a joint state generated by $c$-qubit communication (without prior entanglement) is at most $2^c$. 

\vspace{1em} We next relate the quantum and classical correlation complexities to standard and nonnegative ranks, respectively. 

\begin{Thm}\label{thm:qcorr}
\begin{equation}
	\frac{1}{4}\log_2 \rank(P) \leq \qcorr(p) \leq \min_{Q:\ Q\circ \bar Q = P} \log_2 \rank(Q),
\end{equation}
and the upper bound can be achieved by (local) unitary operations followed by a measurement in the computational basis. 
\end{Thm}
\begin{proof}
\emph{Lower bound:}	Suppose the seed state is $\rho = \mu_i \sum_{i=1}^{2^q}\ket{\psi_i}\bra{\psi_i}$, where $q = 2r = 2\qcorr(p)$ and $\ket{\psi_i}$'s are pure states. Further apply Schmidt decomposition on each $\ket{\psi_i}$: 
\begin{equation}
	\ket{\psi_i} = \sum_{j=1}^{2^r} \lambda_{ij} \ket{\psi_{ij}}\otimes \ket{\phi_{ij}},
\end{equation}
where $\ket{\psi_{ij}}$ and $\ket{\phi_{ij}}$ are in \alice's and \bob's sides, respectively. Now whatever local operations \alice and \bob apply (for generating a distribution $p$) can be formulated as general POVM measurements $\{E_x\}$ and $\{F_y\}$, respectively, resulting in
\begin{align}
	p(x,y) & = \sum_{i} \mu_i \big\langle E_x\otimes E_y, \ket{\psi_i}\bra{\psi_i} \big\rangle = \sum_{ijk} \mu_i \lambda_{ij}\lambda_{ik} \bra{\psi_{ik}} E_x \ket{\psi_{ij}} \cdot \bra{\phi_{ik}} E_y \ket{\phi_{ij}}
\end{align}
Therefore $P$ can be written as the summation of $2^{q+r+r} = 2^{4r}$ rank-1 matrices, \ie $\rank(P) \leq 2^{4r}$.

\vspace{.5em}
\emph{Upper bound:} Consider the singular value decomposition of $Q$: $Q = \sum_{i=1}^r\sigma_i \ket{u_i}\bra{v_i}$, where $\sigma_i>0$, $r = \rank(Q)$, $\ket{u_i}$ and $\ket{v_i}$ are unit length vectors. Observe that 
\begin{align}
	\sum_i \sigma_i^2 & = \|Q\|_F^2 = \sum_{x,y} |Q(x,y)|^2 = \sum_{x,y} Q(x,y) \bar Q(x,y) \\
	& = \sum_{x,y} 	(Q\circ \bar Q)(x,y) = \sum_{x,y} P(x,y) = 1. 
\end{align}
Now let \alice and \bob share the state $\ket{\psi} = \sum_{i=1}^r \sigma_i\ket{i}\otimes \ket{i}$, which is a valid pure state because of the equality above. Then \alice applies $U$ and \bob applies $V$, where $U$ and $V$ are unitary matrices the $i$-th columns of which are $\ket{u_i}$ and $\ket{\bar v_i}$, respectively. Then a measurement in the computational basis gives $(x,y)$ with probability 
\begin{equation}
	\Big|\sum_i \sigma_i \qip{x}{u_i} \qip{y}{\bar v_i}\Big|^2 = \Big|\sum_i \sigma_i \qip{x}{u_i} \qip{v_i}{y}\Big|^2 = |Q(x,y)|^2 = Q(x,y) \bar Q(x,y) = P(x,y), 
\end{equation}
as desired. 
%Suppose the seed state is $\ket{\psi} = \sum_{i=1}^r\sqrt{\lambda_i} \ket{i_A}\ket{i_B}$ in its Schmidt decomposition, where $r = 2^q$ is the Schmidt rank, and $\lambda_i \geq 0$, $\sum_i \lambda_i = 1$. After applying the circuits $U$ and $V$, the state becomes $\sum_{i=1}^r\sqrt{\lambda_i} \ket{u_i}\ket{v_i}$ where $u_i$ and $v_i$ are the $i$-th columns of $U$ and $V$, respectively. Then we measure the state in the computational basis; the probability of observing $(x,y)$ is $|\sum_{i=1}^r \qip{x}{u_i}\qip{y}{v_i}|^2$. Thus $P = Q \circ \bar Q$ where $Q(x,y) = \sum_{i=1}^r \qip{x}{u_i}\qip{y}{v_i}$. But this amounts to saying that $Q$ has rank $r$. 
\end{proof}

An application of the lower bound is to separate the quantum correlation complexity and mutual information, namely, the aforementioned lower bound $I(p)/2 \leq \qcorr(p)$ can be quite loose. 
\begin{Prop}
	There is a correlated distribution $p$ with $I(p) = O(n^{-1/3})$ and $\qcorr(p) \geq \frac{1}{4}\log_2(n+1)$.
\end{Prop}
\begin{proof}
In \cite{HJMR09}, the following distribution is defined to separate mutual information and another two measures $C(p)$ and $T(p)$ which we will not give details but only mention that both are lower bounds for $\rcomm(p)$. The distribution $p$ is defined on $\B^n\times\B^n$:
\begin{equation}
	p(x,y) = \frac{|\{i:x_i = y_i\}|}{n}\cdot 2^{1-2n}
\end{equation}
and they showed that $I(p) = O(n^{-1/3})$. Here we can separate $I(p)$ and $\qcorr(p)$ by showing that $\rank(P) = n+1$ and thus $\qcorr(p) \geq \frac{1}{4}\log_2 (n+1)$. Indeed, consider the submatrix of size $(n+1)\times (n+1)$ where the indices $x,y\in \{0^n, 10^{n-1}, 110^{n-2}, \cdots, 1^n\}$. The submatrix, after a proper scaling, is the following one 
\begin{equation}
\begin{bmatrix} 1 & 1-1/n & 1-2/n & \cdots & 0 \\ 1-1/n & 1 & 1-1/n & \cdots & 1/n \\ 1-2/n & 1-1/n & 1 & \cdots & 2/n \\ \vdots & \vdots & \vdots & \ddots & \vdots \\ 0 & 1/n & 2/n & \cdots & 1\end{bmatrix}.
\end{equation}
By subtracting each row from its next one, it is not hard to see that the rank of this is $n+1$.
\end{proof}
%By tensor $p$ with itself for $k=1/I(p)$ times, we get a separation of $I(p^{\otimes k}) = 1$ and $\qcorr(p^{\otimes k}) = \Omega(n^{1/3}\log n)$? Depending on direct sum of correlated complexities... 

%\paragraph{Remark} Nayak (private communication) gave a separation between mutual information and randomized communication complexity. For any function $f:\mcX\times \mcY \rightarrow \B$ and any product distribution $p$ on $\mcX\times \mcY$, consider the distribution $p'=(X,Y,B)$ where $(X,Y)\leftarrow p$ and $B = f(X,Y)$. Note that $I(X:YB) \leq I(X:Y) + I(X:B|Y) \leq H(B) \leq 1$. But for some function $f$ and some product distribution $p$, $\rcomm_\epsilon(P') = \Omega(\sqrt{n})$, where $P' = [p'(x,y,b)]_{x,(y,b)}$.

\vspace{1em}
Next we fully characterize the randomized correlation and communication complexity by nonnegative rank. The argument of the lower bound was essentially known before (for example, in proving the Cut-and-Paste lemma in \cite{BYJKS04}), here we observe that the argument also yields a lower bound for nonnegative rank. We include it for the completeness. 
\begin{Thm}\label{thm:rcomm}
	$\rcomm(p) = \rcorr(p) = \lceil \log_2 \rank_+(P) \rceil$.
\end{Thm}
\begin{proof}
	We shall prove that $\rcorr(p) \leq \lceil \log_2 \rank_+(P) \rceil$ and $\rcomm(p)\geq \lceil \log_2 \rank_+(P) \rceil$. The conclusion then follows by the bound $\rcomm(p) \leq \rcorr(p)$.
	
	\emph{Upper bound for $\rcorr(p)$}: By definition of $\rank_+(P)$, we can decompose $P$ \st $P(x,y) = K \sum_{i=1}^r q_i a_i(x)b_i(y)$ where $q$, $a_i$'s and $b_i$'s are all probability distributions, $K>0$ is a global normalization factor, and $r = \rank_+(P)$. By summing over all $(x,y)$ and compare the above equality, it is easily seen that actually $K=1$. Therefore, \alice and \bob can sample from $P$ by first sharing a random $i$ distributed according to $q$, and \alice sampling $x$ from $a_i$, \bob sampling $y$ from $b_i$.

\vspace{.5em}
	\emph{Lower bound for $\rcomm(p)$}: Suppose $p$ can be generated by an $r$-round protocol $M = (M_1, \ldots, M_r)$ where the random variable $M_i$ is the message in the $i$-th message. \alice uses private randomness $r_A$ and \bob uses private randomness $r_B$. Without loss of generality, suppose \alice starts the protocol by sending $M_1$. Let $c$ be the total number of bits exchanged. At the end of the protocol \alice outputs $X$ and \bob outputs $Y$. Let $m$ range over the set of possible message. Then
\begin{align}
	p(x,y) = \sum_m \pr_{r_A, r_B}[M = m] \pr_{r_A, r_B}[X = x , Y = y | M = m]
\end{align}
Expand the probability $\pr_{r_A, r_B}[M = m]$ by conditional probabilities in a round-by-round manner, we have	
\begin{align}
	\nonumber \pr_{r_A, r_B}[M = m] & = \pr_{r_A}[M_1 = m_1] \cdot \pr_{r_B}[M_2 = m_2|M_1 = m_1] \cdot \ldots \\
	& \qquad \cdot \pr[M_r = m_r | M_1\ldots M_{r-1} = m_1\ldots m_{r-1}]
\end{align}
where the last probability is over either $r_A$ or $r_B$, depending on the parity of $r$. Finally noting that $\pr_{r_A, r_B}[X = x , Y = y | M = m] = \pr_{r_A}[X = x | M = m] \cdot \pr_{r_B}[Y = y | M = m]$ since conditioned on a fixed message $m$, the \alice and \bob's outputs are independent. Rearranging the terms gives	
\begin{align}
	p(x,y) & = \sum_m \Big(\pr_{r_A}[X = x | M = m] \cdot \prod_{i\in [r]: odd} \pr_{r_A}[M_i = m_i | M_{i-1} = m_{i-1}]\Big)  \\
	& \qquad \quad \cdot \Big(\pr_{r_B}[Y = y | M = m] \cdot \prod_{i\in [r]: even} \pr_{r_B}[M_i = m_i | M_{i-1} = m_{i-1}]\Big) 
\end{align}
Now for each fixed $m$, the first term in the above product depends only on $x$, and the second term depends only on $y$, thus each summand is a rank-1 matrix. Since each entry of the matrix is a product of probabilities, it is also a nonnegative matrix. Thus we have decomposed $P = [p(x,y)]_{x,y}$ into the summation of $2^c$ nonnegative rank-1 matrices, proving the theorem. 
\end{proof}

With the above setup, now we look for matrices $Q$ with small rank and large nonnegative rank for $Q\circ \bar Q$. Consider the following \emph{Euclidean Distance Matrix}: For distinct real numbers $c_1, c_2, \ldots, c_N$ of $\mbR^+$, consider the matrix $Q$ defined by 
\begin{equation}
	Q(x,y) = c_x - c_y
\end{equation} 
for all $x, y \in [N]$. Now we construct our probability distribution matrix $P = [p(x,y)]_{xy}$ by taking the Hadamard product of $Q$ and itself, then normalized: 
\begin{equation}
	P = Q \circ Q/\|Q \circ Q\|_1 = [(c_x-c_y)^2]_{xy}/\|Q \circ Q\|_1
\end{equation} 
Note that $Q$ is a real matrix, so $Q = \bar Q$. Clearly $\rank(Q) = 2$, therefore Theorem \ref{thm:qcorr} implies that $\qcorr(P) = 1$. The classical hardness is immediate from a recently proved result. 
\begin{Thm} [Beasley-Laffey, \cite{BL09}]\label{thm:rank+lb}
	$\rank_+(P) \geq \log_2 N$.
\end{Thm}

By this theorem, we have the separation $\qcorr = 1$ and $\rcorr \geq \log_2(n)$. Letting $n$ go to infinity gives the separation in Theorem \ref{thm:corr separation}. 

Euclidean Distance Matrix is generally formed by taking distinct points $c_i$'s from a $d$-dimensional space, and it is a well-studied subject; see textbook \cite{Dat06} and survey \cite{KW10}. It is also conjectured in \cite{BL09} that actually $\rank_+(P) = N$ for all Euclidean Distance Matrices\footnote{A later paper \cite{LC10} claimed to prove this conjecture. Unfortunately, we think there is a gap in the proof, and after rounds of communications, the authors of \cite{LC10} admitted that the proof was wrong\cite{Chu10}.}. Note that existence of even one Euclidean Distance Matrix with $\rank_+(P) = N$ implies that our separation can be improved to ``1 \emph{vs.} $n$", the largest possible.

\medskip
A final remark is that one can also consider approximate versions of correlation complexity, the minimum seed needed to generate a probability distribution $p'$ which is close to the target $p$. Various distance functions can be considered. Theorem \ref{thm:rcomm} immediately characterizes this quantity as the approximate nonnegative rank, namely the minimum nonnegative rank of a matrix which is close to the given matrix (under the corresponding distance functions). The well-studied approximate nonnegative rank factorization usually uses the Frobenius distance \cite{BBLP07} or total variance ($\ell_1$-distance) \cite{ZFL08,YL10}. 
	
	Shi pointed out the paper \cite{ASTS+03}, the main result of which showed an exponential separation between randomized and quantum communication complexities of approximating a correlation in $\ell_1$-distance.  To be more precise, a natural correlation $p$ of size $n$, arising from the Disjointness function, has $\qcomm_\epsilon(p) = O(\log n\log(1/\epsilon))$, but $\rcomm_\epsilon(p) = \Omega(\sqrt{n})$.  
	
\subsubsection{Conjecture of high nonnegative rank for a random matrix with low \qcorr}\label{sec: conj}
We actually conjecture that a random $P$ with $\qcorr(P) = 1$ and some condition holding has $\rcorr(P) = n$ with probability 1. Let us make the precise statement. 
 
For a matrix $M$, denote by $\bra{m_i}$ the row $i$ and by $\ket{m_j}$ the column $j$. Note that multiplying a whole column by a positive number does not change the rank or the nonnegative rank of a matrix.
\begin{Def}
	A nonnegative matrix $M$ is \emph{in the normal form} if $\|\ket{m_j}\|_1 = 1$ for all $j\in [n]$. A nonnegative factorization $M_{m\times n} = C_{m\times r} D_{r\times n}$ is \emph{in the normal form} if $\|\ket{d_j}\|_1 = 1$ for all $j\in [n]$. A nonnegative factorization $M_{m\times n} = C_{m\times r} D_{r\times n}$ is called \emph{optimal} if $r = rank_+(M)$.
\end{Def}
\begin{Fact}
	Any nonnegative matrix in the normal form has an optimal nonnegative factorization in the normal form.
\end{Fact}
\begin{proof}
	Take an optimal nonnegative factorization $M_{m\times n} = C_{m\times r} D_{r\times n}$. Rewrite it as 
	\begin{equation}
	M = [\ket{c_1}/\|\ket{c_1}\|_1, \cdots, \ket{c_n}/\|\ket{c_n}\|_1] \cdot diag(\|\ket{c_1}\|_1, \cdots, \|\ket{c_n}\|_1)\cdot D.
	\end{equation}
	The middle diagonal matrix can be absorbed into $D$, giving a new matrix $D'$. Now the $\ell_1$ norm of column $j$ of $D'$ is 
	\begin{equation}
		\sum_i \|\ket{c_i}\|_1 d_{ij} = \sum_{ik}c_{ki}d_{ij} = \sum_k m_{kj} = \|\ket{m_j}\|_1 = 1.
	\end{equation}
\end{proof}

We want to define our $Q_n = A_{n\times r}B_{r\times n}$, where $(A,B)$ comes from the following set. 
\begin{equation}
	\mcM_n = \{(A,B): A \in \mbR^{n\times r}, B\in \mbR^{r\times n}, \sum_{i=1}^n \qip{a_i}{b_j} = 1, \forall j\in [n], \text{ and } \qip{a_i}{b_i} = 0, \forall i\in [n]\}.	
\end{equation}
Here the first equality is to make $Q=AB$ in the normal form. The second equality requires that the diagonal entries of $Q_n$ are zero; this is for the purpose of later induction. Let
\begin{equation}
	\mcQ_n = \{Q_n = AB: (A,B) \in \mcM_n\}.	
\end{equation}
One can pick a random $Q\in \mcQ_n$ as follows. First pick random vectors $\bra{a_i}$'s on the unit circle, and then a random $\ket{b_i}$ satisfying the two equalities in the definition of $\mcM_n$. Finally let $Q = AB$.
Our main conjecture is:
\begin{Conj}
	A random $Q \in \mcQ_n$ has $rank_+(Q\circ Q) = n$ with probability 1.
\end{Conj}

\subsection{Correlated equilibrium generation: untrusted local operation model}\label{sec: untrusted}
\begin{figure}%
\begin{center}
\includegraphics[width=3.3in]{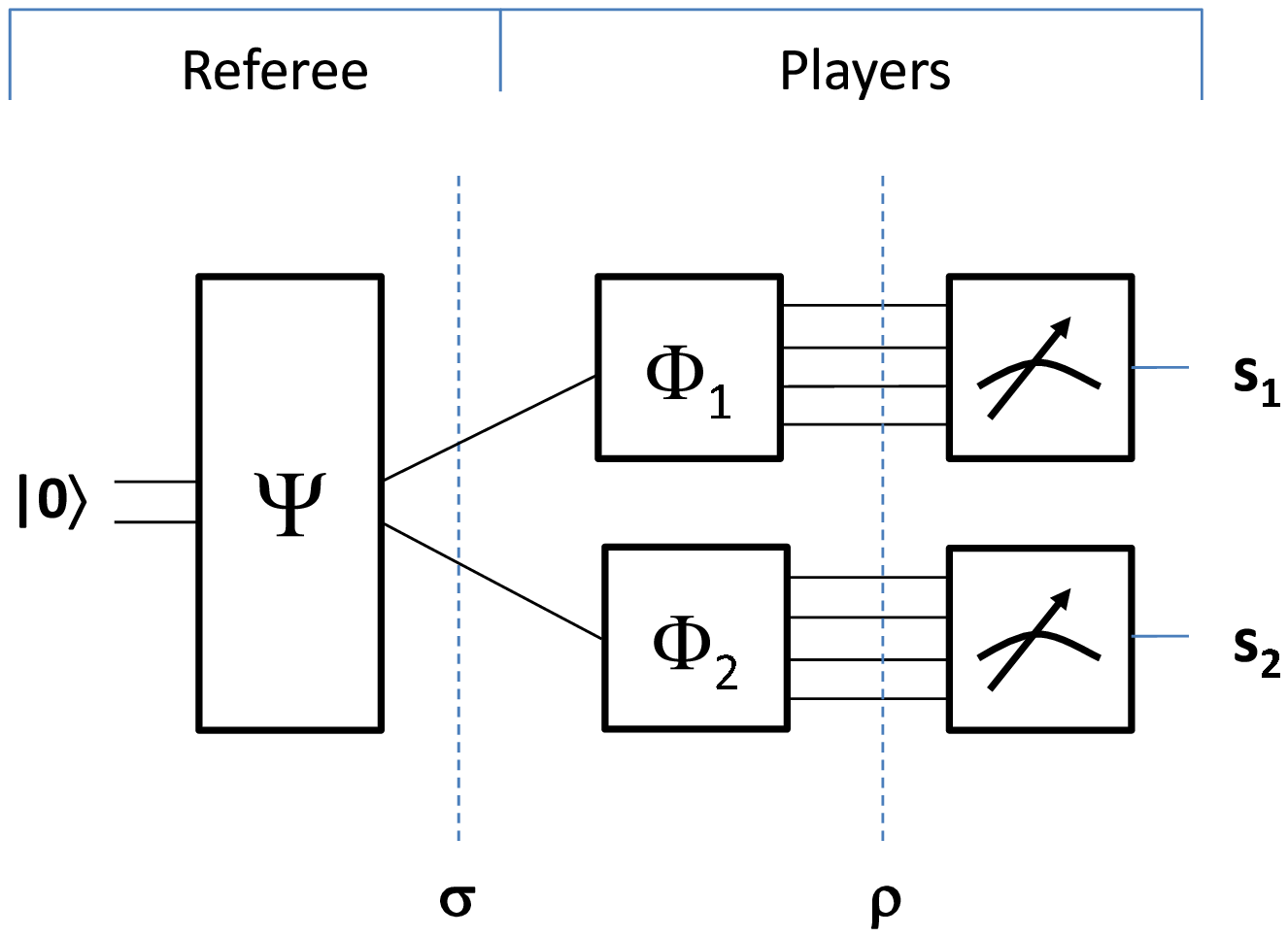}%
\end{center}
\caption{Correlated equilibrium generation with untrusted local operations}%
\label{fig: untrusted}%
\end{figure}

In the untrusted local operation model as illustrated in Figure \ref{fig: untrusted}, the referee generates the seed $\sigma$ and send it to the two players, who then are supposed to finish the correlation generation process by applying the local operations $\Psi_1$ and $\Psi_2$, respectively. However, since the players can deviate from the protocol, the generation process is an equilibrium if no player has incentive to deviate. We define the correlation complexity of generating an CE $p$ in game $G$ as the minimum size of the seed needed. 
\begin{Def}
	The randomized correlation complexity of a distribution $p$ is the minimum size of shared random variables $(X',Y')$ given which 
	\begin{enumerate}
		\item Player 1 and Player 2 can apply local operations $\Phi_1$ and $\Phi_2$ and output $X$ and $Y$, respectively, such that $(X,Y)$ is distributed according to $p$,
		\item no player has incentive to deviate from the protocol, namely, Player $i$ cannot increase her payoff by applying some $\Phi_i'$, provided that the other player does not deviate from the protocol.
	\end{enumerate} The quantum correlation complexity is defined in the same way with the initially shared  $(X',Y')$ being a quantum state. We use $\rcorr(p,G)$ and $\qcorr(p,G)$ to denote the randomized and quantum correlation complexity of CE $p$ in game $G$. 
\end{Def}

It is easy to see that $\rcorr(p,G) \geq \rcorr(p)$ since the definition of $\rcorr(p,G)$ has more requirement than that of $\rcorr(p)$. Similarly we have $\qcorr(p,G) \geq \qcorr(p)$. The following fact is also easy to see because unitary operations are reversible. 
\begin{Fact}
	$\qcorr(p,G) \leq \min_{\ket{\psi}}\qcorr(\ket{\psi})$, where the minimization is over the set \[\{\ket{\psi}: \ket{\psi} \text{ is a QCE of $G$, and $\ket{\psi}$ can be generated by local \emph{unitary} operations (on some seed)}\}.\]
\end{Fact}

Next we show that the separation of classical and quantum correlation generation in the trusted model also applies in the untrusted model for some natural game. Consider the following load-balancing scenario, where each of the two players have $n$ servers to choose. If the two players choose the same server, then both will suffer from the delay due to the collision. If the two players choose different strategies, then they each use one server and no delay is caused, thus they are both happy. So the game matrix is $(J-I, J-I)$. This can be viewed as a natural generalization of the \emph{Traffic Light} game, which is also about collision avoiding. Since in the example in the trusted local operation model, the upper bound of the $\qcorr(p)$ is achieved by unitary operations, and it is easy to verify that the pure state before the measurement is a QCE of the game, the above Fact implies that the same separation also applies here: $\qcorr(p,G) = 1$ and $\rcorr(p,G) \geq \rcorr(p) \geq \log_2 n$. 

A final remark is that though it holds that $\rcorr(p,G) \leq \size(p)$ for all CE $p$, there is no such upper bound for $\rcomm(p,G)$. Actually, for the aforementioned CE $p$ in the Battle of the Sexes game (with half probability on $(A,A)$ and half probability on $(B,B)$), if there is no Referee, then a communication protocol to achieve $p$ actually gives a protocol for weak coin flipping with no bias. This is known to be classically impossible (if no computational assumption is made); in fact in any protocol, there is always one player with success probability being 1. Weak coin flipping with no bias is also impossible for quantum protocols \cite{Amb04}, but the bias can be made arbitrarily close to 0 \cite{Moc07}. This also implies an ``finite" vs. ``infinite" separation between classical and quantum \emph{approximate} communication complexities $\rcomm_\epsilon(p,G)$ and $\qcomm_\epsilon(p,G)$ .

\section{Concluding remarks and open problems}
This work gives a first-step explorations for quantization of classical strategic games, and calls for more systematic studies for quantum strategic game theory. There are lots of problems left open for future work. Some are closely related to quantum games; some are motivated from quantum games but are of their independent interest. 
\begin{enumerate}
	\item {\bf (Maximum increase of payoff)} How to improve the bounds in the first part of Theorem \ref{thm:main1}? Is the maximum quantum incentive in an $n\times n$ bimatrix game always achievable at $A = I_n$, $B = J_n$? Can we give upper bounds better than $1-1/n$ for additive incentive (or better than $n$ for multiplicative incentive) of $\ket{\psi_p}$? Can we solve more low dimensional cases, such as $n=3,4,5$? What is the complexity of finding the maximum quantum incentive for a given bimatrix game? 
	
	\item {\bf (Special games)} There are many important special classes of games, such as zero-sum games, succinctly representable games, etc. It would be interesting to investigate the extremal questions about the maximum incentives in these interesting classes. 
	
	\item {\bf (Average-case games)} How about the increase of payoff for an random game drawn from some natural distribution? 
	
	\item {\bf (Separation between classical and quantum correlation complexities)} Can we improve the separation between randomized and quantum correlation complexities? We conjecture that a random size-$n$ distribution $p$ with $\qcorr(p) = 1$ would have $\rcorr(p) = n$ with probability 1.
%	\item Is it easier to generate a QNE $\rho$ (than an NE)? Note that such a quantum algorithm doesn't guarantee to generate the same QNE on every execution. So we cannot use it to get an approximate NE and argue the hardness as before.

	\item {\bf (Approximate correlation/communication complexity)} Given the connection of approximate randomized correlation/communication complexity and approximate nonnegative rank, can we use the former to answer some questions in the later?
	
	\item {\bf (Characterizing \qcorr)}. We have shown that the randomized correlation and communication complexities are fully characterized by the well-studied measure of the nonnegative rank. Can we have a characterization of the quantum correlation complexity better than the bounds in Theorem \ref{thm:qcorr}? 
	
	\item {\bf (Direct sum/product of correlation and communication complexities)} Do we have direct sum/product for (approximate) correlation and communication complexities?
	
	\item {\bf (Communication complexities of generating CE)} What game $G$ has finite $\rcomm(p,G)$? Has $\rcomm(p,G) = poly(\size(p))$? How about quantum? What if we allow a small error? 
	
	This can be seen as an extension of coin-flipping (without computational assumptions) to the more general case. %While we know that $\rcorr(p) = \rcomm(p)$ and $\qcorr(p) = \qcomm(p)$, the picture is much less clear when when the parties are not trusted. How do we compare $\rcorr(p,G)$ and $\rcomm(p,G)$? And $\qcorr(p,G)$ versus $\qcomm(p,G)$? Is it possible that $\rcomm(p,G) > n$, the size of $p$?
	
	\item {\bf (Testing of quantum equilibria)} How many identical copies of $\rho$ are needed to test the quantum equilibrium property? %What can we say about verification of quantum equilibria?
%	\item Can we improve the bounds in Theorem \ref{thm: superposition}? Can we improve the separation for correlation generation? Improve lower bounds of $\rank_+(EDM)$? Random low rank matrix $Q$ has high $\rank_+(Q\circ \bar Q)$? 
	\end{enumerate} 

\subsection*{Acknowledgments}
The author would like to thank Ming Lam Leung, Yang Li, Cheng Lu for interesting discussions at the early stage of the work, and Andrew Yao for listening to the progress of the work and giving very encouraging comments. An anonymous referee and Dr. Zhaohui Wei helped to improve the presentation of the paper by pointing out some typos. The referee also raised the issue of untrusted player in correlation generating, leading to discussions in Section \ref{sec: untrusted}. Thanks also to Aram Harrow and Andreas Winter for pointing out \cite{Wyn75,Win05}, to Ashwin Nayak and Yaoyun Shi for pointing out \cite{ASTS+03} and giving comment . 

This work was supported by China Basic Research Grant 2011CBA00300 (sub-project 2011CBA00301), and Hong Kong General Research Fund 419309 and 418710. Part of the work was done when the author visited Centre of Quantum Technologies and Tsinghua University, latter under the support of China Basic Research Grant 2007CB807900 (sub-project 2007CB807901).

\bibliography{QCE}
\bibliographystyle{plain}
\end{document}